\documentclass[journal]{IEEEtran}

\usepackage[utf8]{inputenc}
\usepackage{color,soul}
\usepackage[T1]{fontenc}
\usepackage{multicol}
\usepackage{multirow}
\usepackage{textcomp}
\usepackage{bigdelim}
\usepackage{graphicx}
\usepackage{epstopdf}
\usepackage{amssymb}
\usepackage{amsmath}
\usepackage{psfrag}
\usepackage{amsthm}
\usepackage[tight,footnotesize]{subfigure}
\usepackage{cite}
\usepackage{enumerate}
\usepackage{enumitem}
\usepackage{algorithm}
\usepackage{algorithmic}
\usepackage{mathtools}
\usepackage{mathabx}
\usepackage{mathrsfs}
\usepackage{array}
\usepackage{makecell}
\theoremstyle{definition}
\newtheorem{assumption}{Assumption}
\newtheorem{theorem}{Theorem}

\newtheorem{definition}{Definition}

\newtheorem{lemma}{Lemma}

\newtheorem{remark}{Remark}

\makeatletter
\newcommand{\vast}{\bBigg@{3.2}}
\newcommand{\Vast}{\bBigg@{4.5}}

\makeatother

\begin{document}
	
	\title{Trajectory Generation by Chance Constrained Nonlinear MPC with Probabilistic Prediction}

	\author{\IEEEauthorblockN{Xiaoxue Zhang,
		Jun Ma,
		Zilong Cheng,
		Sunan Huang,
		Shuzhi Sam Ge,~\IEEEmembership{Fellow,~IEEE}, and
		Tong Heng Lee}
	\thanks{X. Zhang, Z. Cheng, S. S. Ge, and T. H. Lee are with the NUS Graduate School for Integrative Sciences and Engineering, National University of Singapore, 119077, and the Department of Electrical and Computer Engineering, National University of Singapore, Singapore 117583
		(e-mail: xiaoxuezhang@u.nus.edu; zilongcheng@u.nus.edu; samge@nus.edu.sg; eleleeth@nus.edu.sg).}
	\thanks{J. Ma is with the Department of Mechanical Engineering, University of California, Berkeley, CA 94720 USA (email: jun.ma@berkeley.edu).}
	\thanks{S. Huang is with the Temasek Laboratories, National University of Singapore,  Singapore, 117411 (e-mail: tslhs@nus.edu.sg).}}	
	
	\markboth{IEEE Transactions on Cybernetics}
	{X. Zhang \MakeLowercase{\textit{et al.}}}
	
	\maketitle

\begin{abstract}
Continued great efforts have been dedicated towards high-quality trajectory generation based on optimization methods, however, most of them do not suitably and effectively consider the situation with moving obstacles; and more particularly, the future position of these moving obstacles in the presence of uncertainty within some possible prescribed prediction horizon.
{To cater to this rather major shortcoming,
this work shows how a variational Bayesian Gaussian mixture model (vBGMM) framework can be employed to predict the future trajectory of moving obstacles;} and then with this methodology, a trajectory generation framework is proposed which will efficiently and effectively address trajectory generation in the presence of moving obstacles, and also incorporating presence of uncertainty within a prediction horizon.
{In this work, the full predictive conditional probability density function (PDF) with mean and covariance is obtained}, and thus a future trajectory with uncertainty is formulated as a collision region represented by a confidence ellipsoid. {To avoid the collision region, chance constraints are imposed to restrict the collision probability, and subsequently a nonlinear MPC problem is constructed with these chance constraints.}
It is shown that the proposed approach is able to predict {the future position of the moving obstacles effectively}; and thus based on the environmental information of the probabilistic prediction, it is also shown that the timing of collision avoidance can be earlier than the method without prediction. The tracking error and distance to obstacles of the trajectory with prediction are smaller compared with the method without prediction.
\end{abstract}


\section{Introduction}
\IEEEPARstart{T} {rajectory} generation is
certainly one of the critical component technologies
for autonomous robots~\cite{mora2013teleoperation, huang2019collision,he2018iterative};
and it involves
not only a path planning problem to find a sequence of valid configurations
that moves a mobile robot, but also refers to
the larger problem regarding
how to move along the path in various real-world practical situations.
Pertinent to addressing these mathematical formulations
involving such substantial and possibly difficult
equality and inequality
constraints~\cite{kong2019adaptive,he2020admittance,ma2019parameter,ma2020optimal},
{it is noteworthy that
model predictive control (MPC) is an effective technique
in addressing various constraints as part of
the control synthesis problem~\cite{MPC_IFAC,li2015trajectory,zhou2017real}. However, various drawbacks exist, such as requiring the more restrictive assumption that the unmanned aerial vehicle (UAV) moves on a 2D surface instead of a 3D environment~\cite{NMPC_ICECC}, and lack of consideration of certain environmental factors (obstacles and their motion)~\cite{MPC_IFAC}.}
Overall too,
the prediction of surrounding moving obstacles is a rather challenging problem due to
a large number of factors that influence the future states of robots.

In the existing literature,
various approaches are applied in
a typical trajectory prediction task,
such as Bayesian network~\cite{vgmm2018}, hidden Markov models (HMMs)~\cite{HMM_vehicle2016},
Monte Carlo simulation~\cite{MonteCarlo2019}, Kalman filters~\cite{Kalman2016},
long-short temporal memory (LSTM)~\cite{SocialLSTM_2016,LSTM2019},
generative adversarial networks (GANs)~\cite{assens2018pathgan,roy2019vehicle,SophieCVPR_2019_excellent}, etc.
While all these efforts indicate
great possibilities and promise,
yet at present stages of development,
various drawbacks exist;
such as certain methodologies requiring rather prohibitively high
computational resources
(memory-bandwidth computation)
to train these networks suitably fast,
and also difficulties with
the gap between parameter space and function space.
{Some recent research works, on the other hand, reveal the 
rather significant improvements and advantages with the incorporation and use of a Bayesian network approach~\cite{vgmm2018}}. Here, the probabilistic method gives a probability distribution over the training trajectories, and it additionally provides the conditional distribution of the future horizon given partial history trajectory snippets~\cite{vgmm2012}.
This method also considers a degree of uncertainty for future predictions.

{With all of the above descriptions as a back-drop, in this work, we develop a chance-constrained nonlinear MPC approach to generate the suitable required collision-free trajectory.} We formulate the predicted distribution based on a variational Bayesian Gaussian mixture model (vBGMM) framework as probabilistic chance constraints for the MPC problem; and further solve the resulting nonlinear MPC problem characterizing the collision-free trajectory generation task.

{The key significant contribution of this paper is essentially twofold: both the predicted uncertainty and potential collision are considered
during the prediction horizon in the nonlinear MPC problem.
Therefore with this new and significant development here,
our resulting solution simultaneously ensures that firstly,
the risk of a collision caused by parametric uncertainty and sensor noise is greatly decreased;
and secondly, the required suitable collision-free trajectory can also be generated in advance.
It is noteworthy that
compared to the existing MPC-based methods
without prediction of moving obstacles,
our proposed approach can significantly and effectively improve the quality of the generated trajectories.}

The remainder of this paper is organized as follows.
Section~\ref{section:Prediction} firstly lays out the details of the key basis of our proposed process of trajectory prediction by using the vBGMM framework.
Section~\ref{section:proposed method} then presents the formulation and development of
our proposed trajectory generation approach (with prediction) to efficiently and effectively address trajectory generation in the presence of moving obstacles. Here, uncertainty is incorporated as chance constraints, and an appropriate nonlinear MPC problem is formulated with these constraints.
Then in Section~\ref{section:case study}, a case study on the trajectory generation problem for a quadcopter is given.
Finally, the conclusion of this work is given in Section~\ref{section:conclusion}.

\section{Probabilistic Prediction}
\label{section:Prediction}
The purpose of this section is to show the prediction of the future trajectories for moving obstacles. {Since this probabilistic modelling method requires a probability density function (PDF)}, we can infer the approximated PDF based on the training data. In this section, a joint distribution of history and future data in the training trajectories will be inferred based on vBGMM. Then, the conditional PDF of the future trajectory of test data can be obtained by computing the statistical parameters.

\subsection{Trajectory Representation}
Chebyshev decomposition of trajectories is applied to represent the characteristics of trajectories. The Chebyshev polynomial of a degree of $n$ is defined as~\eqref{eq:Chebyshev}.
\begin{IEEEeqnarray*}{rCl}
	\label{eq:Chebyshev}
	{T_0(x) }& = & {1} \\
	{T_1(x) }& = & {x}\\
	{T_2(x) } & = & {2x^2-1} \\
	& \vdots & \\
	{T_{n+1}(x) } & = & {2xT_n(x) - T_{n-1}(x)}  \yesnumber
\end{IEEEeqnarray*}
The Chebyshev polynomial $T_n$ is orthogonal in the interval $[-1,1]$ and has $n$ zeros in this interval, which means the error between the function we need to approximate and the Chebyshev approximation is close to the optimal $n$th-degree polynomial. To approximate any arbitrary function $f(x)$, the Chebyshev coefficients $a_n$ can be calculated by using \eqref{eq:coefficients}.
\begin{IEEEeqnarray}{ll}
	\label{eq:coefficients}
	a_n = \frac{2}{N} \sum\limits_{k=0}^{N-1}f(x_k)T_n(x_k)
\end{IEEEeqnarray}
where $x_k$ are $N$ zeros of $T_N(x)$. $\boldsymbol a = \begin{bmatrix} a_0 & a_1 & \cdots & a_N \end{bmatrix}$
will be used as input feature to train and predict the probabilistic distribution.
Denote $x,y,z$ as the standard Cartesian coordinates and $v,\theta,\phi$ the spherical coordinates.
For appropriately better capture of
the notation for
the rotation in a trajectory, we use $v,\theta,\phi$ to characterize the trajectory.

\subsection{Variational Bayesian Inference}
The probabilistic trajectory prediction can be formulated as an estimation of the conditional distribution of predicted positions given the history positions of the moving obstacles. This conditional distribution is given by
\begin{IEEEeqnarray}{l}
	\operatorname{Pr}(\boldsymbol{a}_f \mid \boldsymbol{a}_h) = \operatorname{Pr}(a_{v,f}, a_{\theta,f}, a_{\phi,f} \mid a_{v,h},a_{\theta,h}, a_{\phi,h})
\end{IEEEeqnarray}
where $\boldsymbol{a}_h$ and $\boldsymbol{a}_f$ are the Chebyshev approximation coefficients vectors corresponding to history trajectories and future trajectories. All of the subscripts $\cdot_f$ and $\cdot_h$ denote the parameters regarding the future and history, respectively.

First, the joint distribution $\operatorname{Pr}(\boldsymbol{a}_f, \boldsymbol{a}_h)$ can be modeled by GMM which comprises a number of component Gaussian functions to provide a multi-model density function. Some previous researches apply some maximum likelihood solutions or 3-$\sigma$ confidence ellipses to predict the future trajectory~\cite{NMPC+ellipsoid_2017, ellipsoid_2006}. The Bayesian methodology, i.e., variational inference, can be used to estimate this GMM and provide a lower bound on the approximation error~\cite{nasios2006variational}. Variational Bayesian inference has outstanding generation performance and can conquer some shortcomings of these previous methods, such as singularity in the covariance matrix, overfitting, sensibility to the outliers. In this method, the whole conditional predicted distribution can be obtained given the prior distributions of the parameters.
In the Bayesian setting, we consider a prior on the model parameters and aim to infer their posterior distribution as shown in~\eqref{eq:prior}.
\begin{IEEEeqnarray*}{rCl}
	\label{eq:prior}
	\operatorname{Pr}(\boldsymbol{\pi}) & = &{\operatorname{Dir}\left(\boldsymbol{\pi} | {\alpha}_{0}\right)} \label{eq:sub1:prior} \\  \IEEEyesnumber \IEEEyessubnumber
	{} & = & {C\left({\alpha}_{0}\right) \prod_{k=1}^{K} \pi_{k}^{\alpha_{0}-1}} \\
	{\operatorname{Pr}(\boldsymbol{\mu}, \boldsymbol{\Lambda})} & =& {p(\boldsymbol{\mu} | \boldsymbol{\Lambda}) p(\boldsymbol{\Lambda}) } \label{eq:sub2:prior} \\
	{} & =& {\prod_{k=1}^{K} \mathcal{N}\left(\boldsymbol{\mu}_{k} \big| \mathbf{m}_{0},\left(\beta_{0} \boldsymbol{\Lambda}_{k}\right)^{-1}\right) \mathcal{W} \left(\boldsymbol{\Lambda}_{k} | \mathbf{W}_{0}, \nu_{0}\right)} \\ \IEEEnonumber \IEEEyessubnumber
\end{IEEEeqnarray*}
where $K$ is the number of mixture components, $\operatorname{Dir}$ means the Dirichlet distribution, which is used as the conjugate prior of the multinomial distribution of weights $\boldsymbol\pi$, where $\alpha_0$ and $C(\alpha_0)$ are the set of the concentration parameters and the normalization constant of the Dirichlet distribution, respectively. The parameter $\alpha_0$ can be considered as the prior number of observations connected to the components of the mixture model. If the value of $\alpha_0$ is larger, the posterior distribution is more influenced by the prior instead of the data. $\mathcal{N}$ and $\mathcal W$ denote the Normal and Wishart distribution. An independent Normal-Wishart distribution is used as the conjugate prior distribution when both means and precision of Gaussian mixture components $\boldsymbol{\mu}, \boldsymbol{\Lambda}$ are unknown, as shown in~\eqref{eq:sub2:prior}. $\mathbf{W}_{0}, \mathbf{m}_0$ are the initial priors for precisions and means, and $\beta_0, \nu_0$ are the initial scaling factor and degree of freedom of the Wishart distribution, respectively.

It seems infeasible to evaluate the posterior distribution because the dimensionality of the latent space is too high, and the posterior distribution is too complex to have an analytically tractable solution. {Therefore, variational Bayesian inference is useful to obtain the approximated parameters of the posterior distribution. Similar to \cite{vgmm2012}, we also use variational Bayesian expectation-maximization algorithm~\cite{bishop2006pattern} to infer the posterior distribution and obtain the approximated parameters of this distribution.} The predictive density distribution for a new variable $\boldsymbol a$ of the given observed data is a mixture of Student's $t$-distribution~\cite{bishop2006pattern}, which can be calculated by \eqref{eq:pred}.
\begin{IEEEeqnarray}{rCl}
	\label{eq:pred}
	{\operatorname{Pr}(\boldsymbol{a}_f , \boldsymbol{a}_h)} & = & {\frac{\sum_{k=1}^{K} \alpha_{k} T\left(\boldsymbol{a}_f , \boldsymbol{a}_h | \mathbf{m}_{k}, \boldsymbol{L}_{k}, \nu_{k}+1-D\right)} {\sum_{k=1}^{K} \alpha_k}} \nonumber \\
	{\boldsymbol{L}_{k}} & = & {\frac{\left(\nu_{k}+1-D\right) \beta_{k}}{1+\beta_{k}} \mathbf{W}_{k}}
	\yesnumber
\end{IEEEeqnarray}
where $D$ is the dimension of data, $T$ is the Student's $t$-distribution with mean $\boldsymbol{m}_k$ and precision $\boldsymbol{L}_k$ of the $k$th component, $\alpha_k, \beta_k, \nu_k$ are the mixing parameter, scaling factor and degree of freedom of the $k$th component, respectively. The variational lower bound can be used to determine the posterior distribution over $K$ components in the mixture model. A suitable value of $K$ can be determined by treating the mixing coefficients $\boldsymbol\pi$ as parameters and making point estimation by maximizing the lower bound with respect to $\boldsymbol\pi$, rather than computing the distribution by fully Bayesian rule. Hence, re-estimation of the $\boldsymbol\pi$ executes after updating the factorized distribution over other parameters except for $\pi_k$ will lead to sparsity given any initial value of $K$.

At this point, as part of our development to show the prediction of {the future trajectories for moving obstacles}, it is useful to state the following intermediate result on the density probability of the predicted future trajectory (of the observed history trajectory).

\begin{lemma}
\label{lemma1new}
Based on this joint distribution $\operatorname{Pr}(\boldsymbol{a}_f , \boldsymbol{a}_h)$, the density probability of predicted future trajectory of the observed history trajectory can be calculated by computing the conditional distribution $\operatorname{Pr}(\boldsymbol{a}_f \mid \boldsymbol{a}_h)$ as \eqref{eq:condition distr}.
\begin{IEEEeqnarray*}{rCl}
	\label{eq:condition distr}
	\operatorname{Pr}\left(\boldsymbol{a}_{f} | \boldsymbol{a}_{h}\right) & = & \frac{\sum\limits_{k=1}^{K} \tilde{\alpha}_{k} T\left(\boldsymbol{a}_{f} | \boldsymbol{a}_{h}, {\mathbf{m}}_{k, f|h}, {\mathbf{L}}_{k, f|h}, \nu_{k}+1\right)}{\sum_{k=1}^{K} \tilde{\alpha}_k}  \\
	{\tilde{\alpha}_{k}} & {=} & {\frac{\alpha_{k} T\left(\boldsymbol{a}_{h} | \mathbf{m}_{k,h}, \mathbf{L}_{k,h},{\nu_k + 1 -D}\right)}{\sum\limits_{j=1}^{K} \alpha_{j} T\left(\boldsymbol{a}_{h} | \mathbf{m}_{j,h}, \mathbf{L}_{j,h}, {\nu_k + 1 -D}\right)}}  \\
	{{\mathbf{m}}_{k,f|h}} & {=}& {\mathbf{m}_{k,f}+\boldsymbol{\Sigma}_{k,fh} \boldsymbol{\Sigma}_{k, hh}^{-1}\left(\boldsymbol{a}_{h}-\mathbf{m}_{k, h}\right)} \\
	\boldsymbol{L}_{k,f|h}^{-1} &	= & {\left(1+{\left(\boldsymbol{a}_{h}-\mathbf{m}_{k,h}\right)}^{T} \frac{\mathbf{\Sigma}_{k,hh}^{-1}}{{\nu}_{k, f|h}} {\left(\boldsymbol{a}_{h}-\mathbf{m}_{k,h}\right)}\right)} \\
	{} & & \ {\left(\boldsymbol\Sigma_{k,ff}-\boldsymbol\Sigma_{k,fh} \boldsymbol\Sigma_{k,hh}^{-1} \boldsymbol\Sigma_{k,hf}\right)} \frac{{\nu_k + 1 -D}}{{\nu}_{k}-1} \\
	\boldsymbol{\Sigma}_{k,f|h} & = & \frac{{\nu_k -1}}{{\nu}_{k}+1} \boldsymbol{L}_{k,f|h}^{-1}   \yesnumber
\end{IEEEeqnarray*}
where the notation $\cdot_{f|h}$ means the corresponding parameters in the conditional distribution of future data given history data, $\boldsymbol{m}_k = \begin{bmatrix} \boldsymbol{m}_{kh} \\ \boldsymbol{m}_{kf} \end{bmatrix}$, $\boldsymbol{\Sigma}_k = \begin{bmatrix} \boldsymbol{\Sigma}_{k,hh} & \boldsymbol{\Sigma}_{k,hf} \\ \boldsymbol{\Sigma}_{k,fh} & \boldsymbol{\Sigma}_{k,ff} \end{bmatrix}$ are the partition of means and covarainces of this mixture student's $t$-distribution. The subscripts $\cdot_{hh}, \cdot_{hf}, \cdot_{fh}, \cdot_{ff}$ represent the parameter with respect to $\cdot_{\boldsymbol{a}_{h},\boldsymbol{a}_{h}}, \  \cdot_{\boldsymbol{a}_{h},\boldsymbol{a}_{f}}, \  \cdot_{\boldsymbol{a}_{f},\boldsymbol{a}_{h}}, \  \cdot_{\boldsymbol{a}_{f},\boldsymbol{a}_{f}}$, respectively.

\end{lemma}

\begin{proof}
Define $\boldsymbol{X}=\begin{bmatrix}	\boldsymbol{a}_h \\ \boldsymbol{a}_f  \end{bmatrix}$, $\boldsymbol \mu= \begin{bmatrix} \boldsymbol{m}_{h} \\ \boldsymbol{m}_{f} \end{bmatrix}$, $\boldsymbol\Sigma = \begin{bmatrix} \boldsymbol{\Sigma}_{hh} & \boldsymbol{\Sigma}_{hf} \\ \boldsymbol{\Sigma}_{fh} & \boldsymbol{\Sigma}_{ff} \end{bmatrix}$. According to the mixture representation, the characteristic function of $\boldsymbol{X}$ following a multivariate students' $t$-distribution is given by
\begin{IEEEeqnarray*}{rCl}
\label{eq:proofLemma1}
\phi_{\boldsymbol{X}}(\boldsymbol{t}) & = & \mathbb E\left(e^{i \boldsymbol{t}^{T} \boldsymbol{X}}\right) \\
& = & e^{i \boldsymbol{t}^{T} \boldsymbol{\mu}} \frac{\left\|(\nu \boldsymbol{\Sigma})^{1 / 2} \boldsymbol{t}\right\|^{\nu / 2}}{2^{\nu / 2-1} \Gamma(\nu / 2)} K_{\nu / 2}\left(\left\|(\nu \boldsymbol{\Sigma})^{1 / 2} \boldsymbol{t}\right\|\right)  \yesnumber
\end{IEEEeqnarray*}
where $K_{\nu / 2}\left(\left\|(\nu \boldsymbol{\Sigma})^{1 / 2} \boldsymbol{t}\right\|\right)$ is the Macdonald function with order $\nu/2$ and argument $\left\|(\nu \boldsymbol{\Sigma})^{1 / 2} \boldsymbol{t}\right\|$.
Using~\eqref{eq:proofLemma1}, we can obtain $\boldsymbol{a}_h \sim T( \boldsymbol{m}_{h}, \frac{\nu}{\nu-2} \boldsymbol{\Sigma}_{hh})$. Then, the conditional distribution of $\boldsymbol{a}_f$ given $\boldsymbol{a}_h$ is $T( \boldsymbol{m}_{f|h}, \boldsymbol{\Sigma}_{f|h}, \nu_{f|h} )$. Therefore, the conditional distribution of mixture students' $t$-distribution can be written as~\eqref{eq:condition distr}.
\end{proof}

\begin{remark}
	Over-fitting is not a concern when using variational inference as it can find the optimal cluster components $K$ given an initial value.
\end{remark}
The posterior joint distribution and its parameters can be obtained after the training process, and then the conditional distribution will be used to predict the future trajectory in the prediction process based on the parameters calculated in the training process. The derived conditional distribution \eqref{eq:condition distr} defines a conditional PDF of the future trajectories whose mean and covariance can be evaluated by
\begin{IEEEeqnarray*}{rCl}
	\boldsymbol{\mu} & = & \begin{bmatrix} \boldsymbol{\mu}_v & \boldsymbol{\mu}_{\theta} & \boldsymbol{\mu}_{\phi}   \end{bmatrix} = {\sum_{k=1}^{K} \tilde{\alpha}_{k} \boldsymbol{m}_{k}} \\
	{\boldsymbol{\Sigma}} & = & \
	\begin{bmatrix} \boldsymbol{\Sigma}_{v,v} & \boldsymbol{\Sigma}_{v,\theta} & \boldsymbol{\Sigma}_{v,\phi} \\
		\boldsymbol{\Sigma}_{\theta,v} & \boldsymbol{\Sigma}_{\theta,\theta} & \boldsymbol{\Sigma}_{\theta,\phi} \\
		\boldsymbol{\Sigma}_{\phi,v} & \boldsymbol{\Sigma}_{\phi,\theta} & \boldsymbol{\Sigma}_{\phi, \phi}  \end{bmatrix}  \\
	& = & {\sum_{k=1}^{K} \tilde{\alpha}_{k}\left(\boldsymbol{\Sigma}_{k}+\left(\boldsymbol{m}_{k}-\boldsymbol{m}\right)\left(\boldsymbol{m}_{k}-\boldsymbol{m}\right)^{T}\right)} \yesnumber
\end{IEEEeqnarray*}
After training, the predicted Chebyshev coefficients are distributed with $\boldsymbol{a}_{(\cdot)} \sim \mathcal N(\boldsymbol{\mu}_{(\cdot)}, {\boldsymbol{\Sigma}_{(\cdot),(\cdot)}})$, where $\boldsymbol{\mu}_{(\cdot)}$ and ${\boldsymbol{\Sigma}_{(\cdot),(\cdot)}}$ are corresponding mean and covariance for each variable $v,\theta,\phi$. Thus, the mean and covariance of $v,\theta,\phi$ can be evaluated by reconstructing this Chebyshev approximation based on the coefficients $\boldsymbol{a}_{(\cdot)}$. {Then, it is followed by a transformation function from spherical coordinate $v,\theta,\phi$ to Cartesian coordinate $x,y,z$.}

\section{Nonlinear MPC with Chance Constraints}
\label{section:proposed method}
In this section, a nonlinear MPC problem is formulated and then solved appropriately. First, for the purpose of collision avoidance, the obstacle region is represented by ellipsoids and further transformed into chance constraints. Then, we reformulate the collision-avoidance chance constraints as deterministic constraints, which are integrated into the MPC problem. Then, the stability analysis is provided, and an optimization algorithm is presented to solve this problem.
\subsection{Obstacle Region}
\label{section:ellipsoid}
After probabilistic trajectory prediction, we can obtain the means and covariances of the future trajectory, which can be formulated as a predicted region where the host agent needs to avoid, called obstacle region $\mathcal I$. For $i$th moving obstacle, assume its future position probabilistically lies in the obstacle region $\mathcal I_i$ at time $t$ which is based on the mean $\boldsymbol{\mu}_i(t)=\begin{bmatrix} \mu_x & \mu_{y} & \mu_{z} \end{bmatrix}^T$ and covariance $\boldsymbol{\Sigma}_i(t) \in \mathbb R^{3 \times 3}$ with respect to time $t$. For simplicity, the variable $t$ will be neglected for the following description in this section. In such way, we can assume that the future position of the $i$th moving obstacle can be represented as $\boldsymbol{p}_{i,f}\sim \mathcal N(\boldsymbol{\mu}_i, \boldsymbol{\Sigma}_i)$ at time $t$.
\begin{remark}
	Since the covariance matrix $\boldsymbol{\Sigma}_i$ is real symmetric and positive semi-definite, the eigenvalues are real, and there exists an orthogonal matrix $\mathbf{Q}_i$ formed by eigenvectors of $\boldsymbol{\Sigma}_i$, we can carry out the spectral decomposition for the covariance matrix $\boldsymbol{\Sigma}_i$ as
	\begin{IEEEeqnarray}{rCl}
		\label{eq:eigendecomposition}
		{\boldsymbol{\Sigma}_i} & = & {\mathbf{Q}_i \boldsymbol{\Lambda}_i \mathbf{Q}_i^{T}}
	\end{IEEEeqnarray}
	where $\boldsymbol{\Lambda}_i=\operatorname{diag}(\lambda_j),\ j = 1,2,3$, where $\lambda_j$ is sorted in descending order with $\lambda_1 \geq \lambda_2 \geq \lambda_3$. Here, $j$ means the each dimension in the environment.
\end{remark}

At this juncture, it is pertinent to state
two key intermediate results
(on ellipsoid construction, and on approximate scaling factor computation)
that are significant essential parts
in the development which follows {the nonlinear MPC with chance constraints methodology}.
Thus firstly, note the following first
intermediate result
on ellipsoid construction.
\begin{lemma}
\label{lemma2}
	 Ellipsoid can be constructed from the transformation of a unit sphere by firstly stretching with a ratio of $\sqrt{\lambda}_i$ along each axis, then rotating the ellipsoid by $\mathbf{Q}_i$ and a final translation of distribution center $\mathbf m_i$ according to the following inverse Mahalanobis transformation.
	 \begin{IEEEeqnarray}{rCl}
	 	\label{eq:transformation}
	 	{\mathcal{I}_i} & = & {\mathbf{Q}_i \boldsymbol{\Lambda}_i^{\frac{1}{2}} \mathbf{Q}_i^{T} \mathbf{u} +\boldsymbol{\mu}_i}
	 \end{IEEEeqnarray}
	 where $\mathbf{u}\sim \mathcal N(0, \boldsymbol I_3)$ is in a unit sphere with normal distribution in 3 dimensions. In this work, $\boldsymbol I_n$ denotes the identity matrix with the size of $n\times n$.
\end{lemma}
\begin{proof}
	Mapping a unit sphere by the square root of the covariance matrix, $\boldsymbol{\Sigma}_i^{\frac{1}{2}}$ determines an ellipsoid whose principle semi-axes rely on the eigenvalues of this matrix and the orientation is related to the corresponding eigenvectors. In order to represent this ellipsoid graphically, Mahalanobis transformation can be used to eliminate the correlation between the variables and to standardize each variable with variance~\cite{Mahalanobis}. Therefore, a ellipsoid can be constructed from the transformation of a unit sphere, according to the inverse Mahalanobis transformation~\eqref{eq:transformation}.
\end{proof}
Next here, note the following
intermediate result
on
approximate scaling factor computation.

\begin{lemma}
\label{lemma3}
	The approximate scaling factor $r$ can be computed by
	\begin{IEEEeqnarray*}{rCl}
		\label{eq:scale factor}
		F(r) & = & \mathscr P(r) - \tilde{\varphi} \\
		\dot F(r) & = & \dot {\mathscr P}(r) = \sqrt{\frac{2}{\pi}}\exp{\left(-\frac{r^2}{2}\right)} + \frac{\exp{\left(-\frac{r^2}{2}\right)}}{\sqrt{2}\Gamma\left(\frac{3}{2}\right)}(r^2-1) \label{eq:derivative} \\ \IEEEyesnumber
	\end{IEEEeqnarray*}
\end{lemma}
\begin{proof}
	Based on Lemma~\ref{lemma2}, the square of Mahalanobis distance (scaling factor $r$) of the probable position $\boldsymbol{p}_{i,f}$ to its mean $\boldsymbol{\mu}_i$ can be calculated by
	\begin{IEEEeqnarray}{rCl}
		\label{eq:distance}
		{r^2} & = & {(\boldsymbol{p}_{i,f} - \mathbf{m}_i)^T \boldsymbol{\Sigma}_i^{-1} (\boldsymbol{p}_{i,f} - \mathbf{m}_i)}
	\end{IEEEeqnarray}
	Substituting the \eqref{eq:eigendecomposition} and \eqref{eq:transformation} into \eqref{eq:distance}, we can obtain that the magnified ellipsoid with ratio $r$ relys on the chi-square $\chi^2$ distribution with a degree of freedom $\varrho=3$, as shown in \eqref{eq:dist4ellpsoid1}.
	\begin{IEEEeqnarray}{rCl}
		\label{eq:dist4ellpsoid1}
		{\operatorname{Pr}(r^2 \leq \chi^2_{\varrho=3,p})} & = & {\tilde{\varphi}}
	\end{IEEEeqnarray}
	which can be represented by
	\begin{IEEEeqnarray*}{rCl}
		\label{eq:dist4ellpsoid2}
		{\operatorname{Pr}\left((\boldsymbol{a}_{i,f} - \boldsymbol{\mu}_i)^T \boldsymbol{\Sigma}_i^{-1} (\boldsymbol{a}_{i,f} - \boldsymbol{\mu}_i) \leq r^2\right)} & = & {\operatorname{Pr}(\mathbf{u}^T \mathbf{u}\leq r^2)}
	\end{IEEEeqnarray*}
	
	The confidence probability for an arbitrary ellipsoid with any factor $r$ is
	\begin{IEEEeqnarray*}{rCl}
		{\mathscr P(r)} & = & {\operatorname{Pr}(\mathbf{u}^T\mathbf{u} \leq r^2)} \\
		{} & = & {\iiint (2\pi)^{-\frac{3}{2}} \exp{\left(\frac{u_{1}^2+u_{2}^2+u_{3}^2}{2}\right)} du_{1}du_{2}du_{3}} \\
		{} & = & {\operatorname{erf}\left(\frac{r}{\sqrt{2}}\right) - {\left(\frac{r}{\sqrt{2}}\right)}\frac{\exp{\left(-\frac{r^2}{2}\right)}}{\Gamma(\frac{3}{2})} }  \yesnumber
	\end{IEEEeqnarray*}
	where
	\begin{IEEEeqnarray*}{rCl}
		\operatorname{erf}(x) &=& \frac{2}{\sqrt{\pi}} \int_0^x \exp(-t^2) dt
	\end{IEEEeqnarray*}
	is the standard error function, $\Gamma$ is the gamma function. 
	Given the confidence level, the scaling factor of the ellipsoid can be calculated by cumulative distribution function $F(r)$ and its derivatives $\dot{F}(x)$ as~\eqref{eq:scale factor}.
\end{proof}
Then, $r$ can be solved by iterative Newton-based methods based on Lemma~\ref{lemma3}. In fact, the confidence ellipsoids in different confidence levels can be obtained, which form the obstacle region $\mathcal I_i$ for the $i$th moving obstacle. In such case, the scaling factor $r = 2.5003, 2.7955, 3.3682$ when the confidence level $p=90\%, 95\%, 99\%$, respectively.

\subsection{Chance Constraint}
Assume there are $n_o$ moving obstacles the host robot can detect at the moment $t$. Checking whether there is collision happening between the host robot and a moving obstacle $i$ requires to compute the minimum distance between the current position of robot $\boldsymbol{p}(t)$ and the collision region of the $i$th moving obstacle $\mathcal I_i$. Notably, $\boldsymbol{p}(t)$ is part of the state variable $\boldsymbol{x}(t)$. The collision condition of the host robot with respect to the moving obstacle $i$ at time $t$ is defined as
\begin{IEEEeqnarray}{rcl}
	\mathbb C_i^t := \{\boldsymbol p(t) \mid \|\boldsymbol p(t) - \hat{\boldsymbol{p}}_i(t) \| \leq d_{\text{safe}}, \quad \hat{\boldsymbol{p}}_i(t) \in \mathcal{I}_i(t) \}
\end{IEEEeqnarray}
where $\hat{\boldsymbol{p}}_i(t)$ denotes the possible position of the $i$th obstacle in time $t$, $d_\textup{safe}$ means the predefined safety distance between the host agent and the moving obstacles, and $\|\cdot\|$ is the Euclidean norm. This condition means if $\boldsymbol{p}(t) \in \mathbb C_i^t$ is satisfied, there might be a collision happened between the host robot and the $i$th moving obstacle. Since the predicted positions are represented by a probability distribution, the predicted collision avoidance constraints can be formulated in a probabilistic manner, which are so-called chance constraints:
\begin{IEEEeqnarray}{rCl}
	\label{eq:ori_constraint}
	{\operatorname{Pr}\left(\boldsymbol{p}(t) \in \mathbb C_i^t \right)} & \leq & { \varphi, \ i\in \mathbb N_{n_o}}
\end{IEEEeqnarray}
where $\varphi$ is the probability threshold for the robot-obstacle collision, the set $\mathbb N_{n_o}=\{1,2,\cdots, n_o\}$ and $n_o$ is the number of moving obstacles the robot can detect.
At this point, it is
pertinent
to also state the following
intermediate result.
\begin{lemma}
\label{lemma4}
Given any matrix $\mathbf{A}$ and scalar $b$, for a multivariate random variable $\mathbf X(t)$ corresponding to the mean $\boldsymbol\mu(t)$ and covariance $\boldsymbol\Sigma(t)$, the chance constraint
\begin{IEEEeqnarray*}{rCl}
	\label{eq:chance_constraint}
	{\operatorname{Pr}\left(\mathbf{A}^T \mathbf{X}(t) < b \right)}  \leq \varphi   \yesnumber
\end{IEEEeqnarray*}
is equivalent to a deterministic linear constraint
\begin{IEEEeqnarray*}{rCl}
\label{eq:chance_constraint_2}
	{\mathbf{A}^T \boldsymbol\mu(t) - b} \ge \eta  \yesnumber
\end{IEEEeqnarray*}
where $\eta = {\sqrt{2\mathbf{A}^T \boldsymbol\Sigma(t) \mathbf{A}}\operatorname{erf}^{-1}(1-2\varphi)}$ and $\varphi$ is the predefined allowable probability threshold of collision.
\end{lemma}

\begin{proof}
Given a univariate Gaussian random variable $X\sim \mathcal {N}(\mu, \sigma^2)$ with known variance, according to the definition of PDF, we have that $\operatorname{Pr}(X<0) \leq \varphi$ is equal to $\mu \ge \eta$, where $\eta=\sqrt{2}\sigma\operatorname{erf}^{-1}(1-2\varphi)$.

In terms of multivariate Gaussian random variable $X(t)\sim \mathcal N(\mu(t), \Sigma(t))$ at time $t$, set a univariate random variable $Y(t)$ is the perpendicular distance between the plane $\mathbf{A}^T \mathbf{X}(t) = b$ and the point $X(t)$, and then the event $\mathbf{A}^T \mathbf{X}(t) < b$ is equal to $Y(t)<0$. Based on the relationship between $Y(t)$ and $X(t)$, we have $Y(t)\sim \mathcal N(\mu_Y, \sigma_Y)$, where $\mu_Y = \mathbf A^T\mu(t)-b$ and $\sigma_Y=\sqrt{\mathbf A^T \Sigma(t) \mathbf A}$. Here, $\operatorname{Pr}(\mathbf{A}^T \mathbf{X}(t) < b)\leq \varphi$ is equal to $\operatorname{Pr}(Y(t)<0)\leq \varphi$. Applying the above result of univariate Gaussian random variable $Y(t)$, we can obtain the $\mu_Y\geq \eta$, where $\eta=\sqrt{2}\sigma_Y\operatorname{erf}^{-1}(1-2\varphi)$. Therefore, \eqref{eq:chance_constraint} is equivalent to \eqref{eq:chance_constraint_2}.
\end{proof}
Particularly key in our work here
is the appropriately interesting utilization
of the notion of chance constraints,
where (as also mentioned earlier)
the predicted collision avoidance constraints can be formulated as chance constraints.
Along this line then,
the following main result
is of particular importance.

\begin{theorem}\label{thm:1}
	The chance constraint~\eqref{eq:ori_constraint} can be reformulated as a deterministic linear constraint as
	\begin{IEEEeqnarray*}{lll}
		\label{eq:new chance constraint}
		{\boldsymbol{\kappa}_i^T(t) \left(\boldsymbol{p}(t)-\Pi_{\mathcal I_i(t)}(\boldsymbol{p}(t))\right)}  \\
		{\quad \quad  \geq  {\sqrt{2\boldsymbol{\kappa}_i^T(t) \boldsymbol\Sigma(t) \boldsymbol{\kappa}_i(t)}\operatorname{erf}^{-1}(1-2\varphi)}} \yesnumber
	\end{IEEEeqnarray*}
where
\begin{IEEEeqnarray*}{l}
	\boldsymbol{\kappa}_i(t) = \frac{\boldsymbol{p}(t)-\Pi_{\mathcal I_i(t)}(\boldsymbol{p}(t))}{\|\boldsymbol{p}(t)-\Pi_{\mathcal I_i(t)}(\boldsymbol{p}(t))\|}
\end{IEEEeqnarray*}
is the slope of the line connecting $\boldsymbol{p}(t)$ and $\Pi_{\mathcal I_i(t)}(\boldsymbol{p}(t))$ and is perpendicular to the tangent plane.
\end{theorem}

\begin{proof}
	As mentioned in Section~\ref{section:ellipsoid}, the obstacle region $\mathcal I_i(t)$ for any moving obstacle $i$ at each time moment $t$ can be described as an ellipsoid according to its mean $\boldsymbol{\mu}_i(t)$ and covariance $\boldsymbol{\Sigma}_i(t)$. Obviously, this ellipsoid region $\mathcal I_i(t)$ is a convex set. We can find a point $\hat{\boldsymbol{p}} \in \mathcal I_i$ which is the closest point from the given position $\boldsymbol{p}(t)$ of host robot.
	The closest point can be represented in the projection form
	\begin{IEEEeqnarray}{l}
		\hat{\boldsymbol{p}}(t)  =  \Pi_{\mathcal I_i(t)}(\boldsymbol{p}(t))
	\end{IEEEeqnarray}
	where $\Pi_{\mathcal I_i(t)}(\boldsymbol{p}(t)) := \min\{\frac{1}{2}\| \hat{\boldsymbol{p}}(t)-\boldsymbol{p}(t)\|^2 \mid \hat{\boldsymbol{p}}(t) \in \mathcal{I}_i\}$ is the projection of $\boldsymbol{p}(t)$ onto $\mathcal I_i(t)$.
	
	Then, the tangent plane over the closest point $\Pi_{\mathcal I_i(t)}(\boldsymbol{p}(t))$ is perpendicular to the line from $\boldsymbol{p}(t)$ to $\Pi_{\mathcal I_i(t)}(\boldsymbol{p}(t))$, which can be represented by
	\begin{IEEEeqnarray*}{l}
		\boldsymbol{\kappa}_i^T(t) \left(\boldsymbol{p}(t)-\Pi_{\mathcal I_i(t)}(\boldsymbol{p}(t))\right) = 0 \yesnumber
	\end{IEEEeqnarray*}
	
	Therefore, the collision region can be enlarged as a half space
	\begin{IEEEeqnarray}{l}
		\label{eq:constr relax}
		\hat{\mathbb C}_i(t) := \{\boldsymbol{p} \mid \boldsymbol{\kappa}_i^T(t) \left(\boldsymbol{p}(t)-\Pi_{\mathcal I_i(t)}(\boldsymbol{p}(t))\right) \leq 0 \}
	\end{IEEEeqnarray}
	It's obvious that the collision region $\mathcal I_i \subset \hat{\mathbb C}_i(t)$, and thus ${\operatorname{Pr}\left(\boldsymbol{p}(t) \in \mathcal I_i(t) \right)} \leq {\operatorname{Pr}\left(\boldsymbol{p}(t) \in \hat{\mathbb C}_i(t) \right)} $. The original chance constraints \eqref{eq:ori_constraint} can be relaxed as \eqref{eq:constr relax} and reformulated as deterministic linear constraints \eqref{eq:new chance constraint}, based on Lemma~\ref{lemma4}.
\end{proof}

\begin{remark}
	The closest point $\Pi_{\mathcal I_i(t)}(\boldsymbol{p}(t))$ can be calculated by solving the following optimization problem:
	\begin{IEEEeqnarray*}{rcc}
		\operatorname*{min}_{\boldsymbol{y}}  & \quad {(\boldsymbol{y}-\boldsymbol{p}_{i,f})^T(\boldsymbol{y}-\boldsymbol{p}_{i,f})} \\
		\operatorname{s.t.}  & \quad {(\boldsymbol{y} - \boldsymbol{\mu}_i)^T \boldsymbol{\Sigma}_i^{-1} (\boldsymbol{y} - \boldsymbol{\mu}_i) \leq r^2} \yesnumber
	\end{IEEEeqnarray*}
	Obviously, if $\boldsymbol{p}_{i,f}$ is inside $\mathcal I_i$, then $ \boldsymbol{y} = \Pi_{\mathcal I_i(t)}(\boldsymbol{p}(t))  = \boldsymbol{p}_{i,f}$ and  distance between the closest point and position of host agent $\operatorname{dist}(\Pi_{\mathcal I_i(t)}(\boldsymbol{p}(t)), \boldsymbol{p}_{i,f})=0$. Otherwise $\Pi_{\mathcal I_i(t)}(\boldsymbol{p}(t)) = \boldsymbol{y}$ is on the boundary of $\mathcal I_i$. This problem can be transformed into a quadratic minimization problem, and thus the Lagrangian method can be used~\cite{ma2017integrated,ma2017novel}. Define the Lagrangian function
	\begin{IEEEeqnarray*}{rCl}
		\mathcal L & = & (\boldsymbol{y}-\boldsymbol{p}_{i,f})^T(\boldsymbol{y}-\boldsymbol{p}_{i,f}) + \\
		{} & {} & \lambda \left((\boldsymbol{y} - \boldsymbol{\mu}_i)^T \boldsymbol{\Sigma}_i^{-1} (\boldsymbol{y} - \boldsymbol{\mu}_i) - r^2 \right) \yesnumber
	\end{IEEEeqnarray*}
	where $\lambda$ is the Lagrange multiplier.
	The Karush-Kuhn-Tucker (KKT) conditions are
	\begin{IEEEeqnarray*}{rCl}
		\label{eq:kkt}
		\frac{\partial{\mathcal L}}{\partial \boldsymbol{y}} & = & 2(\boldsymbol{y}-\boldsymbol{p}_{i,f})+2\lambda \boldsymbol{\Sigma}_i^{-1} (\boldsymbol{y} - \mathbf{m}_i) = 0 \\
		\frac{\partial{\mathcal L}}{\partial \lambda} & = & (\boldsymbol{y} - \boldsymbol{\mu}_i)^T \boldsymbol{\Sigma}_i^{-1} (\boldsymbol{y} - \boldsymbol{\mu}_i) - r^2 = 0 \yesnumber \\
	\end{IEEEeqnarray*}
	Hence, the optimal solution $\boldsymbol{y}^*$, i.e., $\Pi_{\mathcal I_i(t)}(\boldsymbol{p}(t))$ can be obtained by solving~\eqref{eq:kkt} via gradient-based methods.
\end{remark}

\subsection{Problem Formulation}
\label{section:general problem}
Based on the trajectory prediction and chance constraints reformulation, we can interpret the probabilistic prediction as deterministic linear constraints. Therefore, we can formulate an MPC problem to find the (sub-)optimal control input sequence to generate (sub-)optimal trajectory for the host robot, while considering the obstacles' future positions.
\subsubsection{Dynamic Model}
Here, we consider a nonlinear dynamic model for an agent, which can be written as \eqref{eq:dynamics_general}.
\begin{IEEEeqnarray}{rCl}
	\label{eq:dynamics_general}
	\boldsymbol{x}_{k+1}=f(\boldsymbol x_{k}, \boldsymbol u_{k}), \ k \in \mathbb{N}_{N-1}
\end{IEEEeqnarray}
where $\boldsymbol{x}_k \in \mathbb R^n$ and $\boldsymbol{u}_k \in \mathbb R^m$ denote the state variables and control inputs of this dynamic model at time step $k$, $f$ represents the dynamics, $\mathbb{N}_{N-1}$ is the set of non-negative integers. This nonlinear dynamic model can be approximated by a linear time-variant system model with time-variant matrices ${\boldsymbol A}_t\in \mathbb R^{n\times n}$ and ${\boldsymbol B}_t\in \mathbb R^{n\times m}$.
\subsubsection{Constraints}
Some physical limitations need to be considered when computing the optimal control inputs, where are shown in \eqref{eq:constr_bound_set}.
\begin{IEEEeqnarray}{rCl}
\label{eq:constr_bound_set}
	\boldsymbol{x}_{k+1}\in \mathcal X, \boldsymbol{u}_k \in \mathcal U
\end{IEEEeqnarray}
where $\mathcal X\in \mathbb R^{n}$ and $\mathcal U\in \mathbb R^m$ denote the bounded set of $\boldsymbol{x}$ and $\boldsymbol{u}$, respectively.

In order to avoid the potential collision with obstacles, the reformulated chance constraints can be embedded as part of the constraints of MPC problem, as \eqref{eq:new chance constraint}.

For other static obstacles, the generated trajectory should guarantee that the distance between current location and obstacles is greater than the predefined safe distance $d_{\text{safe}}$, as shown in \eqref{constr:general dist_static}.

\begin{IEEEeqnarray}{c}
	\label{constr:general dist_static}
	\left\|L\boldsymbol{x}_k-L{\boldsymbol x}_{\mathrm{obs}}^k \right\| \geq d_{\text{safe}}
\end{IEEEeqnarray}
where $L$ is a linear operator to take out the position vector from state vector $\boldsymbol{x}$, ${\boldsymbol x}_{\mathrm{obs}}^k$ is the position coordinates of the nearest $i$th obstacle within the detection radius of this agent.

\subsubsection{Cost Function}
Define the stage cost function as

\begin{IEEEeqnarray*}{rCl}
	 \label{eq:stage_cost_general}
	\ell({\boldsymbol x}_k, {\boldsymbol u}_k) &=&  \left\| {\boldsymbol x}_k\right\|^2_{\boldsymbol Q} + \left\| {\boldsymbol u}_k\right\|^2_{\boldsymbol R} \IEEEyesnumber
\end{IEEEeqnarray*}
where $N$ is the prediction horizon, $\left\|{\boldsymbol x}_k\right\|^2_{\boldsymbol Q}={\boldsymbol x}_k^T{\boldsymbol Q}{\boldsymbol x}_k$, $\left\| {\boldsymbol u}_k\right\|^2_{\boldsymbol R} = {\boldsymbol u}_k^T{\boldsymbol R}{\boldsymbol u}_k$, $\boldsymbol Q \in \mathbb R^{n\times n}$ and $\boldsymbol R \in \mathbb R^{m\times m}$ are weighting matrices, and $\boldsymbol{Q} \succ 0, \boldsymbol{R}\succ 0$. The terminal cost is
\begin{IEEEeqnarray}{ll}
	\label{eq:terminal cost}
	\ell_{f}(\boldsymbol x_N)= \left\| {\boldsymbol x}_N \right\|^2_{\boldsymbol P}
\end{IEEEeqnarray}
where $\left\|{\boldsymbol x}_N\right\|^2_{\boldsymbol P}={\boldsymbol x}_N^T{\boldsymbol P}{\boldsymbol x}_N$,  $\boldsymbol{P}\in \mathbb R^{n\times n}$ is  the penalty matrix, and $\boldsymbol{P} \succ 0$.
Then, we have the cost function is
\begin{IEEEeqnarray}{rCl}
\label{eq:cost_added}
	J(\boldsymbol x(t), U(t)) &=& \sum_{k=0}^{N-1} \ell\left(\boldsymbol{x}_{t+k|t}, \boldsymbol{u}_{t+k|t}\right)+\ell_{f}\left(\boldsymbol{x}_{t+N|t}\right) \IEEEeqnarraynumspace
\end{IEEEeqnarray}
where ${\boldsymbol{U}}(t)=\begin{bmatrix} \boldsymbol u_{t|t}^T & \boldsymbol u_{t+1|t}^T & \cdots & \boldsymbol u_{t+N-1|t}^T \end{bmatrix}^T \in \mathbb R^{Nm}$ is the sequence of control inputs over the prediction horizon $N$.

Therefore the following nonlinear MPC problem can be formulated
as a constrained finite horizon nonlinear quadratic optimal control problem at time $t$ as~\eqref{eq:generalProblem}.

\begin{IEEEeqnarray*}{rCl}
	\label{eq:generalProblem}
	\operatorname*{min}_{\boldsymbol{U}} & \quad & J(\boldsymbol x(t), U(t)) \\
	\operatorname {s.t.} &\quad & \boldsymbol{x}_{t+k+1|t}=f(\boldsymbol x_{t+k|t}, \boldsymbol u_{t+k|t}), k \in \mathbb{N}_{N-1} \\
	 &\quad& \boldsymbol{x}_{t|t}=\boldsymbol{x}(t) \\
	 &\quad& \boldsymbol{x}_{t+k|t}\in \mathcal X, \boldsymbol{u}_{t+k|t}\in \mathcal U \\
	 &\quad& {\boldsymbol{\kappa}_{i,k}^T \left(\boldsymbol{x}_{t+k|t}-\Pi_{\mathcal I_{i,k}}\left(\boldsymbol{x}_{t+k|t}\right)\right)}  \\
	 &\quad&  \quad \geq  {\sqrt{2\boldsymbol{\kappa}_{i,k}^T \boldsymbol\Sigma_{i,k} \boldsymbol{\kappa}_{i,k}}\operatorname{erf}^{-1}(1-2\varphi)} \\
	 &\quad& ||L\boldsymbol x_{t+k|t}-L{{\boldsymbol x}_{\mathrm{obs}}}_{t+k}|| \geq d_{\text{safe}} \\
	 &\quad& \boldsymbol{x}_N \in \mathcal X_f \yesnumber
\end{IEEEeqnarray*}
where the terminal constraint region $\mathcal X_f$ is a polytope.

\subsection{Stability Analysis}
\label{subsubsection:stability_analysis}
Ahead of the stability analysis suitably characterizing the performance of the proposed methodology,  the following assumptions, definitions, and lemmas are introduced to
provide a sufficient condition to prove the stability.
Here, we focus on the conditions for the uniform asymptotical stability of the origin of system~\eqref{eq:dynamics_general} under the cost function~\eqref{eq:cost_added}, physical limitation constraint~\eqref{eq:constr_bound_set} and terminal region $\mathcal X_f$, with the feedback control law.
\begin{assumption}
\label{assumption1}
	Consider the problem~\eqref{eq:generalProblem}. If the initial problem at time $0$ is feasible, then the problem at time $t$ is feasible for all $t>0$. (It may be noted that this is a reasonable assumption pertaining to the posed physical system being a typical actual system where actual real-world solutions can be admissible. This should be expected of typical actual systems.)
\end{assumption}
\begin{definition}
\label{definition:postive-definite}
	A continuous function $V(t,\boldsymbol{x})$ is a locally positive definite function, if $V(t,0)=0$ and
	\begin{IEEEeqnarray}{rC}
	V(t,\boldsymbol{x})\geq \gamma(|\boldsymbol{x}|), & \forall \boldsymbol{x}\in \mathbb B_d, t\geq 0	
	\end{IEEEeqnarray}
	where $\mathbb B_d$ is a ball centred in the origin with radius $d$, and the function $\gamma(\cdot)$ is continuous and strictly increasing with $\gamma(0)=0$.
\end{definition}
\begin{definition}
\label{definition:decrescent}
	A function $V(t,\boldsymbol{x})$ is decrescent with $\boldsymbol{x}\in \mathcal X$, if there exists a function $\epsilon(\cdot)$ such that $V(t,\boldsymbol{x}(t)) \leq \epsilon(|\boldsymbol{x}|)$, where function $\epsilon(\cdot)$ is continuous and  strictly increasing with $\epsilon(0)=0$.
\end{definition}

\begin{lemma}
\label{lemma_imp}
	Consider the system~\eqref{eq:dynamics_general} with $L\boldsymbol{x}(t)=\boldsymbol{p}(t)$, and let a continuous function $V_N(t,\boldsymbol{x}(t))$ be stated as the value function. The origin of system is locally uniformly asymptotically stable, if $V_N(t,\boldsymbol{x}(t))$ is a local positive definite function.
\end{lemma}

\begin{proof}
	The proof of this result stated above is shown in~\cite{sastry2013nonlinear}.
\end{proof}
\begin{lemma}
\label{lemma_prove}
 For the system~\eqref{eq:dynamics_general}, with the stage cost function $\ell(\cdot)$ and terminal cost function $\ell_f(\cdot)$, the value function $V_N(t,\boldsymbol{x}(t))$ is a decrescent function in the domain $\mathcal X$ for $\boldsymbol{x}(t)\in \mathcal X$.
\end{lemma}
\begin{proof}
The value function is
\begin{IEEEeqnarray*}{rCl}
\label{eq:value_func}
V_N(t,\boldsymbol{x}(t)) & = & \sum\limits_{k=0}^{N-1}\left( \left\|\boldsymbol{x}_{t+k|t} \right\|^2_{\boldsymbol{Q}} + \left\|\boldsymbol{u}_{t+k|t}\right\|^2_{\boldsymbol{R}} \right) + \left\|\boldsymbol{x}_{t+N|t}\right\|^2_{\boldsymbol{P}} \\ \yesnumber
\end{IEEEeqnarray*}
which can be rewritten as
\begin{IEEEeqnarray}{rCl}
	V_N(t,\boldsymbol{x}(t)) & = &
	\begin{bmatrix} {\boldsymbol{x}}^T(t) & {\boldsymbol{U}}^T(t) \end{bmatrix}
	\boldsymbol{\Xi}
	\begin{bmatrix} {\boldsymbol{x}}(t) \\ {\boldsymbol{U}}(t) \end{bmatrix}
\end{IEEEeqnarray}
 where the matrix $\boldsymbol{\Xi}\in \mathbb R^{(n+Nm)\times(n+Nm)}$ can be determined by $\boldsymbol{Q}, \boldsymbol{R}, \boldsymbol{P}, \boldsymbol{A}_t, \boldsymbol{B}_t$. The components of ${\boldsymbol{u}}(t)$ are bounded since it is a solution of problem~\eqref{eq:generalProblem} and the set $\mathcal U$ is compact. The same property can also apply to ${\boldsymbol{x}}(t)$. The partial derivative of $f(\boldsymbol{x}, \boldsymbol{u})$ with respect to $\boldsymbol{x}$ and $\boldsymbol{u}$ are bounded, and thus matrices $\boldsymbol{A}_t,\boldsymbol{B}_t$ are bounded. Therefore, $V_N(\cdot,\cdot)$ is bounded and we can always find a positive definite function $\epsilon(\cdot)$ such that $V_N(t,\boldsymbol{x}(t)) \leq \epsilon(\boldsymbol{x}), \forall \boldsymbol{x} \in \mathcal X$~\cite{bemporad2002explicit}. According to Definition~\ref{definition:decrescent}, the value function $V_N(t,\boldsymbol{x}(t))$ is a decrescent function in the domain $\mathcal X$ for $\boldsymbol{x}(t)\in \mathcal X$. This completes the proof of the result.
\end{proof}

In the following, we use $V_N(t,\boldsymbol{x}(t))$ as a Lyapunov function to find a sufficient condition to prove the stability.
\begin{theorem}
\label{theorem_imp}
	Consider the system~\eqref{eq:dynamics_general} with the terminal constraint $\mathcal X_f=0$, physical limitation constraint~\eqref{eq:constr_bound_set} and the cost function~\eqref{eq:cost_added}. Let $\Delta V_N(t,\boldsymbol{x}(t))=V_N(t-1,\boldsymbol{x}(t-1))-V_N(t,\boldsymbol{x}(t))$, then the function $\Delta V_N(\cdot,\boldsymbol{x}(\cdot))$ is a locally positive definite function if
	\begin{IEEEeqnarray*}{lCl}
	\label{eq:stability_condition}
	\ell(\boldsymbol{x}_{t+N-1|t},\boldsymbol{u}_{t+N-1|t}) < \ell(\boldsymbol{x}_{t-1|t-1},\boldsymbol{u}_{t-1|t-1})
	\\
	\quad - \sum\limits_{i=0}^{N-2} \left\|\boldsymbol{x}_{t+i|t}-\boldsymbol{x}_{t+i|t-1})\right\|^2_{\boldsymbol{Q}}\yesnumber
	\end{IEEEeqnarray*}
	Then it follows that the origin of the closed-loop system is uniformly, locally asymptotically stable.
\end{theorem}

\begin{proof}
	Recall that the value function at time $t$ is
	\begin{IEEEeqnarray*}{rCl}
		V_N(t,\boldsymbol{x}(t)) &=& \sum\limits_{k=0}^{N-1} \left( \left\|\boldsymbol{x}_{t+k|t}\right\|^2_{\boldsymbol{Q}} + \left\|\boldsymbol{u}_{t+k|t}\right\|^2_{\boldsymbol{R}}\right) + \left\| \boldsymbol{x}_{t+N|t}^T \right\|^2_{\boldsymbol{P}} \\ \yesnumber
	\end{IEEEeqnarray*}
	where $\boldsymbol{u}_{t+k,t}$ with $k=0,1,\cdots, N-1$ is the solution of the optimization problem at time $t$. The state of the system at time $t$ is $\boldsymbol{x}(t)=f(\boldsymbol{x}(t-1),u_{t-1|t-1})$. At time $t-1$, we have the solution ${\boldsymbol{U}}_{t-1} = \begin{bmatrix} u_{t-1|t-1}^T & u_{t|t-1}^T & \cdots  & u_{t+N-2|t-1}^T \end{bmatrix}^T$. Set the sequence $\hat{\boldsymbol{U}}_{t} = \begin{bmatrix} u_{t|t-1}^T & u_{t+1|t-1}^T & \cdots  & u_{t+N-2|t-1}^T & u_{t+N-1|t}^T \end{bmatrix}^T$ which is feasible for the problem, and is obtained from ${\boldsymbol{U}}_{t-1}$ by removing $u_{t-1|t-1}$ and adding $u_{t+N-1|t}$. 
	The function $\Delta V_N(t,\boldsymbol{x}(t))$ can be rewritten as
	\begin{IEEEeqnarray*}{rCl}
	\label{eq:condition_prove}
		&& \Delta V_N(t,\boldsymbol{x}(t)) \\
		&=& \sum\limits_{k=0}^{N-1} \left\| \boldsymbol{x}_{t+k-1|t-1} \right\|^2_{\boldsymbol{Q}} + \sum\limits_{k=0}^{N-1} \left\|\boldsymbol{u}_{t+k-1|t-1}\right\|^2_{\boldsymbol{R}}\\
		&&  + \left\|\boldsymbol{x}_{t+N-1|t-1}\right\|^2_{\boldsymbol{P}} - \sum\limits_{k=0}^{N-1}\left\|\boldsymbol{x}_{t+k|t}\right\|^2_{\boldsymbol{Q}} - \sum\limits_{k=1}^{N-1} \left\|\boldsymbol{u}_{t+k-1|t-1}\right\|^2_{\boldsymbol{R}}\\
		&&  - \left\|\boldsymbol{u}_{t+N-1|t}\right\|^2_{\boldsymbol{R}} - \left\|\boldsymbol{x}_{t+N|t}\right\|^2_{\boldsymbol{P}}\\
		&=& \sum\limits_{k=0}^{N-2} \left(\left\|\boldsymbol{x}_{t+k|t-1} \right\|^2_{\boldsymbol{Q}} - \left\|\boldsymbol{x}_{t+k|t} \right\|^2_{\boldsymbol{Q}} \right) + \left\|\boldsymbol{x}_{t-1|t-1}\right\|^2_{\boldsymbol{Q}} \\
		&&  + \left\|\boldsymbol{u}_{t-1|t-1}\right\|^2_{\boldsymbol{R}} - \left\|\boldsymbol{x}_{t+N-1|t}\right\|^2_{\boldsymbol{Q}} - \left\|\boldsymbol{u}_{t+N-1|t}\right\|^2_{\boldsymbol{R}}\\
		&\geq & -\sum\limits_{k=0}^{N-2} \left\|\boldsymbol{x}_{t+k|t} - \boldsymbol{x}_{t+k|t-1} \right\|^2_{\boldsymbol{Q}} \\
		&& - \ell\left(\boldsymbol{x}_{t+N-1|t},\boldsymbol{u}_{t+N-1|t}\right) + \ell\left(\boldsymbol{x}_{t-1|t-1},\boldsymbol{u}_{t-1|t-1}\right)  \yesnumber
	\end{IEEEeqnarray*}
	If the condition~\eqref{eq:stability_condition} holds, $\Delta V_N(t,\boldsymbol{x}(t))$ is a locally positive definite function since the right hand side of~\eqref{eq:condition_prove} is positive and bounded. Then, the function $\Delta V_N(t,\boldsymbol{x}(t))$ is a locally positive definite function and it is a decrescent function based on Lemma~\ref{lemma_prove}. According then to Lemma~\ref{lemma_imp}, the origin of the system is uniformly asymptotically stable.
\end{proof}

\begin{remark}
	The condition~\eqref{eq:stability_condition} is established for the stated nonlinear system above, and it directly leads to an additional convex constraint to be embedded and incorporated in the required MPC design and algorithm. With all of these in place at this point, this then is an appropriate applicable stability result for the methodology proposed here.
\end{remark}

\subsection{Proposed Algorithm}
\label{section:problem solving}
At time $t$, the cost function is optimized under the constraints in~\eqref{eq:generalProblem} to obtain the optimal control sequence $\boldsymbol{U}(t)$.
It is worthwhile to mention that only the first control input $\boldsymbol u_{t|t}$ will be executed. Multiple shooting method~\cite{multiple_shooting} is used to solve this nonlinear optimization problem with multiple constraints~\eqref{eq:generalProblem}. First, discretize the system dynamics and constraints at each time $t$ over a coarse discrete time grid $k=0,1,\cdots,N$ with sampling step $\Delta t$. For each time $t$, a boundary value problem is solved with imposing some additional continuity constraints. 
This problem can be expressed as a nonlinear program which can be solved using sequential quadratic programming. The interior-point method or the active set method can be applied to solve the corresponding quadratic program. 

Here, Algorithm~\ref{alg:NMPC} is used to obtain the optimal control inputs at time $t$, and thus the optimal control sequence can be obtained. If there is no solution of this problem but solving time $t_\text{comp}$ is less than the predefined maximum computation time limits $t_{\max}$, a nonnegative slack vector $\Theta_k \in \mathbb R^{n_x}$ will be added to soften the inequality constraints~\cite{kerrigan2000soft}, as shown in \eqref{eq:relax_ineq_constr}.
\begin{IEEEeqnarray*}{rCl}
\label{eq:relax_ineq_constr}
\boldsymbol H_x \boldsymbol x_{t+k|t} & \leq & \boldsymbol h_x + \Theta_k \\
\Theta_k & \succeq & 0 \yesnumber
\end{IEEEeqnarray*}
where the symbol $\succeq$ means element-wise no less than 0, the subscript $\cdot_k$ denotes each step in prediction horizon of the MPC problem with $k\in \mathbb N_{N-1}$, and $\boldsymbol H_x, \boldsymbol h_x$ are arbitrary given matrix and vector. The value of the slack vector relies on the degree of associated acceptable violation of the constraints.

The slack vector $\tilde{\Theta}_k$ can also be added in the equality constraints to transform the equality constraints into tube-like constraints, as shown in \eqref{eq:relax_eqconstr}.
\begin{IEEEeqnarray*}{C}
\label{eq:relax_eqconstr}
f(\boldsymbol x_{t+k|t}, \boldsymbol u_{t+k|t}) - \tilde{\Theta}_k \leq  \boldsymbol{x}_{t+k+1|t} \leq f(\boldsymbol x_{t+k|t}, \boldsymbol u_{t+k|t}) + {\tilde{\Theta}}_k \\
{\tilde{\Theta}}_k  \succeq  0 \yesnumber
\end{IEEEeqnarray*}

Moreover, the slack vector $\hat{\Theta}_k$ can also be added to the cost function as a scalar weight to ensure that the slacking is not abused~\cite{kerrigan2000soft}, as shown in \eqref{eq:relax_cost}.
\begin{IEEEeqnarray*}{rCl}
\label{eq:relax_cost}
\operatorname*{minimize}_{\boldsymbol x} &&\quad \sum_{k=0}^{N-1} \left( \left\|\boldsymbol x_{t+k|t}\right\|^2_{\boldsymbol Q} +\left\|\boldsymbol u_{t+k|t}\right\|^2_{\boldsymbol R} +\rho \hat{\Theta}_{k}^{T} \hat{\Theta}_{k}\right) \\
&& + \left\|\boldsymbol x_{t+N|t}\right\|^2_{\boldsymbol P}  \yesnumber
\end{IEEEeqnarray*}
where $\rho$ is the penalty of the slack vector.
\begin{remark}
Adding slacking vectors on control input constraints is not reasonable, as the inputs often originate from an actuator which has hard limits constraining the force, torque etc.
\end{remark}

If there is no feasible solution or the computation time to solve this problem $t_\text{comp}$ exceeds the maximum computation time limits $t_{\max}$, a backup controller will be invoked, thereby continuing the control progress. For example, a conservatively tuned PID controller can be used which sacrifices performance for relaxed constraints satisfaction. An alternative way is to execute the control inputs at the upper/lower bound.

\begin{algorithm}[!thbp]
	\caption{Numerical procedures for nonlinear MPC problem at time $t$ with chance constraints} \label{alg:NMPC}
	\begin{algorithmic}[1]
		\STATE Initialize state vector of the agent $\boldsymbol x_0$; Initialize the number of repeats $n_\textup{rep} = 1$; Set both of the maximum number of repeats $n_\textup{set}$ and the maximum computation time $t_{\max}$ to the reasonable values.
		\STATE Compute the initial value of the cost function by \eqref{eq:cost_added}.
		\FOR {$k=0, 1, \cdots, N-1$}
		\STATE Compute $\boldsymbol x_{t+k+1|t}$ based on $\boldsymbol x_{t+k|t}$ using \eqref{eq:dynamics_general}.
		\STATE Generate the constraints of the physical limitations by \eqref{eq:constr_bound_set}.
		\STATE Generate the constraints of collision avoidance to moving obstacles by \eqref{eq:new chance constraint}.
		\STATE Generate the constraints of collision avoidance to static obstacles by \eqref{constr:general dist_static}.
		\STATE {Generate the additional constraint based on the stability condition~\eqref{eq:stability_condition}.}
		\ENDFOR
		\STATE Compute the cost function by \eqref{eq:cost_added}.
		\STATE Solve the nonlinear MPC problem \eqref{eq:generalProblem}.
		\IF {There exists a solution}
		\RETURN control sequence $\boldsymbol U(t)$
		\ELSE
		\STATE Add slack vectors on these constraints and cost function.
		\STATE $n_\textup{rep} = n_\textup{rep} + 1$.
		\STATE Go to line 11.
		\ENDIF
		\STATE Compute the computation time collapse $t_\text{comp}$.
		\IF {$n_\textup{rep}>n_\textup{set}$ or $t_\text{comp}>t_{\max}$}
		\STATE Turn on the backup controller.
		\ENDIF
	\end{algorithmic}
\end{algorithm}

\section{Case Study}
\label{section:case study}
Among all of the applications of trajectory generation, a case study on a UAV system would certainly be an ideal test platform for the 3D trajectory generation problem, since trajectory planning typically need to work in 3-dimensional state space with multiple degrees of freedom and multiple constraints due to its dynamical characteristics and physical limits in a typical UAV application. To demonstrate the effectiveness of our proposed method, a UAV (quadcopter) system is used as an application test platform for trajectory generation in this section.

\subsection{Plant Model}
\label{section:case_model}
The world coordinate system $\mathbb{W}$ and the robot body coordinate system $\mathbb{B}$ are shown in Fig.~\ref{fig:model}, where $x_W, y_W, z_W$ are three dimensions in the world-fixed frame and $x_B, y_B, z_B$ are in the body-fixed frame. As shown in Fig.~\ref{fig:model}, each quadcopter is equipped four rotors. For each rotor, there are a vertical force due to the rotation of the rotor and a moment perpendicular to the plane of the propeller rotation. Therefore, there are four vertical force $F_1, F_2, F_3, F_4$ to overcome the gravity and drive the quadcopter.
The dynamic model of a quadcopter can be represented by \eqref{eq:dyn} with neglecting the aerodynamic and gyroscopic effects~\cite{zhang2019integrated}.
\begin{figure}[h]
	\centering
	\includegraphics[width=0.32\textwidth]{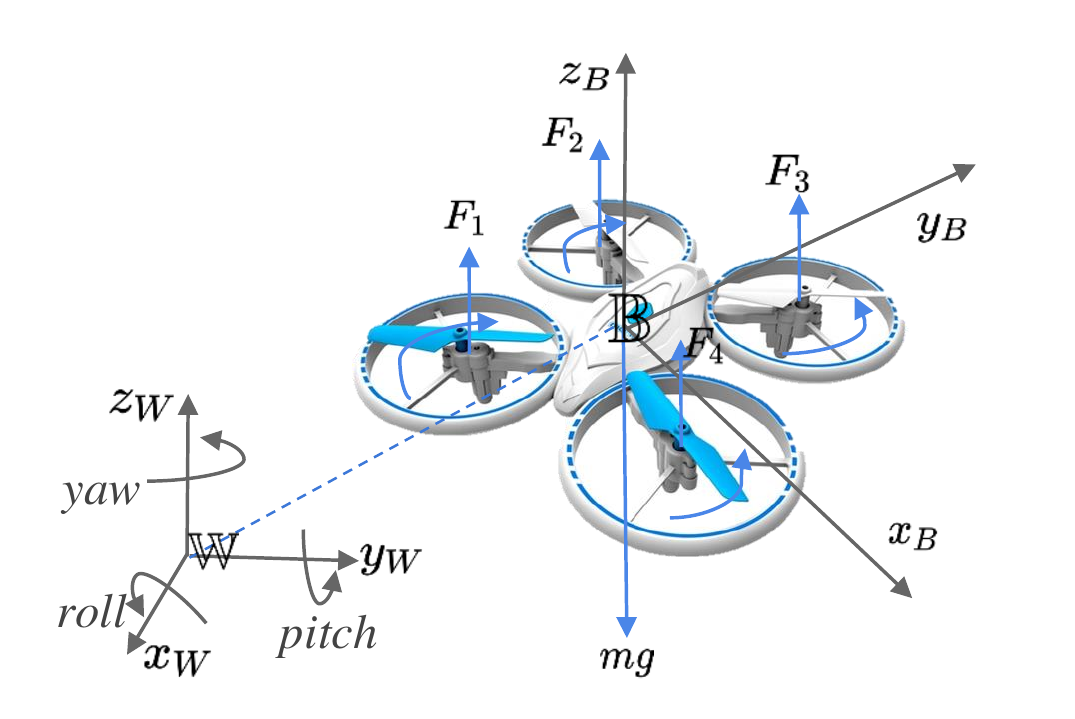}
	\caption{\label{fig:model} Model of the quadcopter.}
\end{figure}

\begin{IEEEeqnarray}{cl}
	\label{eq:dyn}
	\IEEEyesnumber \IEEEyessubnumber*
	{\dot{\boldsymbol p}(t)}  & {=\boldsymbol{R}(\phi, \theta, \psi)^T {\boldsymbol v}(t)} \label{subeq1:dyn}\\
	{\dot{\boldsymbol v}(t)}  & {=-{\boldsymbol\omega}(t) \times {\boldsymbol v}(t) +g\boldsymbol R(\phi, \theta, \psi){\boldsymbol e} + {\boldsymbol e}T/m} \label{subeq2:dyn} \\
	{\dot{\boldsymbol\zeta}(t)} & {=\boldsymbol W(\psi, \theta, \phi) {\boldsymbol\omega}(t)} \label{subeq3:dyn} \\
	{\dot{\boldsymbol\omega}(t)} & {=\boldsymbol J^{-1}(-{\boldsymbol\omega}(t)\times \boldsymbol J{\boldsymbol\omega}(t) + {\boldsymbol\tau})} \label{subeq4:dyn}
\end{IEEEeqnarray}
Here, \eqref{subeq1:dyn} models the position of quadcopter in the world coordinates ${\boldsymbol p}=\begin{bmatrix} p_x & p_y & p_z \end{bmatrix}^T \in \mathbb R^3$; \eqref{subeq2:dyn} is to represent the velocity of the quadcopter in three dimensions ${\boldsymbol v}=\begin{bmatrix} v_x & v_y & v_z\end{bmatrix}^T \in \mathbb R^3$; $\boldsymbol\zeta=\begin{bmatrix} \phi& \theta & \psi \end{bmatrix}^T \in \mathbb R^3$ denotes three angles, roll, pitch and yaw, respectively; The angular velocity in three dimensions is represented by $\boldsymbol\omega=\begin{bmatrix} \omega_x & \omega_y & \omega_z\end{bmatrix}^T \in \mathbb R^3$; ${\boldsymbol e}=\begin{bmatrix} 0 & 0 &1\end{bmatrix}^T \in \mathbb R^3$; ${\boldsymbol\tau}=\begin{bmatrix} \tau_x & \tau_y & \tau_z \end{bmatrix}^T \in \mathbb R^3$ represents the torques of the quadcopter in each dimension; $g$ is the gravitational acceleration; $m$ is the mass of this quadcopter; $T$ denotes the total thrust; $\boldsymbol J=\textup{diag}(J_x, J_y, J_z) \in \mathbb R^{3\times 3}$ denotes the moment of inertia of the quadcopter; and $\boldsymbol R(\psi, \theta, \phi)\in \mathbb R^{3\times 3}$ denotes the rotation matrix of the quadcopter
(fuller details in~\cite{zhang2019integrated}).

The rotor thrusts of the four rotors are chosen as control inputs, i.e., $\hat{\boldsymbol u} = \begin{bmatrix} F_1 & F_2 & F_3 & F_4\end{bmatrix}^T \in \mathbb R^{4}$, and then we have the relationship between individual thrusts and individual torques which is expressed by \eqref{eq:control_input}.
\begin{IEEEeqnarray}{ll}
	\label{eq:control_input}
	\begin{bmatrix}
		{T} \\ {\tau_{x}} \\ {\tau_{y}} \\ {\tau_{z}}
	\end{bmatrix} =
	\begin{bmatrix}
		{-1} & {-1} & {-1} & {-1} \\
		{0} & {-L} & {0} & {L} \\
		{L} & {0} & {-L} & {0} \\
		{-c} & {c} & {-c} & {c}
	\end{bmatrix}
	\begin{bmatrix}
		{F_{1}} \\ {F_{2}} \\ {F_{3}} \\ {F_{4}}
	\end{bmatrix}
\end{IEEEeqnarray}
where $L$ is the distance from the rotor to the center of gravity of the quadrotor and $c$ is a constant that relates the rotor angular momentum to the rotor thrust.

Define the state vector as
\begin{IEEEeqnarray}{cl}
	\IEEEnonumber
	{\boldsymbol x} &{=\begin{bmatrix} p_x \ \  p_y \ \  p_z  \ \  v_x  \ \  v_y \ \  v_z  \ \  \phi \ \   \theta \ \   \psi  \ \   \omega_x \ \    \omega_y \ \   \omega_z \end{bmatrix}^T }\\
	&{= \begin{bmatrix} \boldsymbol p^T & \boldsymbol v^T & \boldsymbol\zeta^T & \boldsymbol\omega^T \end{bmatrix}^T \in \mathbb R^{12}}
\end{IEEEeqnarray}

The dynamic model of a quadcopter can be formulated by the form of $\dot{\boldsymbol x}=f(\boldsymbol x)+\tilde{\boldsymbol B} \hat{\boldsymbol u}$. Let $\hat{\boldsymbol u}=\boldsymbol u_{\mathrm{eq}}+ \boldsymbol u$ with $\hat{\boldsymbol u} \in \mathbb R^4$, where $\boldsymbol u_{\mathrm{eq}} = \begin{bmatrix} \frac{mg}{4} & \frac{mg}{4} & \frac{mg}{4} & \frac{mg}{4} \end{bmatrix}^T$ is used to overcome the gravity of the quadcopter. Therefore, the quadcopter system can be represented by~\eqref{eq:model}.
\begin{IEEEeqnarray*}{c}
	\label{eq:model}
	{\begin{bmatrix} {\dot{\boldsymbol p}} \\ {\dot{\boldsymbol v}} \\ {\dot{\boldsymbol\zeta}} \\ {\dot{\boldsymbol\omega}}\end{bmatrix}} =
	{\begin{bmatrix} {\boldsymbol R(\phi, \theta, \psi)}^T {\boldsymbol v} \\ {-{\boldsymbol\omega}\times \boldsymbol v + g {\boldsymbol R(\phi, \theta, \psi)} \boldsymbol e} \\ \boldsymbol W(\phi, \theta, \psi) {\boldsymbol\omega} \\ \boldsymbol J^{-1}(-{\boldsymbol\omega} \times {\boldsymbol J} {\boldsymbol\omega}) \end{bmatrix}
		+ {\tilde{\boldsymbol B}} {({\boldsymbol u}_{\mathrm{eq}} + {\boldsymbol u})}}  \yesnumber
\end{IEEEeqnarray*}
Due to constraints of space, more details of the matrix ${\tilde{\boldsymbol B}} \in \mathbb R^{12\times 4}$
refer to~\cite{zhang2019integrated}.

\subsection{Model Linearization}
The expression
\eqref{eq:model} is nonlinear and time-variant,
so the state-dependent coefficient factorization~\cite{SDE}
is used to handle and address the nonlinear dynamics.
The resulting state-space expression can
be expressed by~\eqref{eq:state-space}.

\begin{IEEEeqnarray}{rCl}
	\label{eq:state-space}
	{\boldsymbol x_{k+1|t}}=(\tilde{\boldsymbol A}_t\Delta t + I) \boldsymbol x_{t+k|t} + \tilde{\boldsymbol B} \Delta t (\boldsymbol u_\textup{eq}+ \boldsymbol u_{t+k|t}) \IEEEeqnarraynumspace
\end{IEEEeqnarray}
where $\Delta t$ is the sampling time interval.
Since ${\tilde{\boldsymbol A}}_t$ and $\tilde{\boldsymbol B}$
are dependent on the current state $\boldsymbol x$,
this state-space representation is a pseudo-linear form,
and then we can suitably consider the system matrices to be constant during the prediction horizon.
The full details of the matrix $\tilde{\boldsymbol A}_t \in \mathbb R^{12}$ are shown in Appendix~\ref{section:appendix_matrix}.
The control design focuses on $\boldsymbol u _{t+k|t}$. When $\boldsymbol u _{t+k|t} = 0$, the quadcopter lies in an equilibrium situation as $\tilde{\boldsymbol A}_t \boldsymbol x_{t+k|t} + \tilde{\boldsymbol B} \boldsymbol u_\textup{eq} = 0$.

\subsection{Problem Formulation}
According to the aforementioned analysis, a nonlinear MPC problem with chance constraints at time $t$ can be formulated as \eqref{eq:MPC problem}.

\begin{IEEEeqnarray*}{cl}
	\label{eq:MPC problem}
	{\operatorname*{min}_{\boldsymbol{u}}} & \quad \sum_{k=0}^{N} \left\| \boldsymbol x_{t+k|t} - {{\boldsymbol x}_\text{ref}}_{t+k} \right\|_{\boldsymbol Q} + \sum_{k=0}^{N-1} \left\| \boldsymbol u_{t+k|t} \right\|_{\boldsymbol R} \\
	{\operatorname{s.t.}} & {\quad {\boldsymbol x_{k+1|t}}=(\tilde{\boldsymbol A}_t\Delta t + I) \boldsymbol x_{t+k|t} + \tilde{\boldsymbol B} \Delta t (\boldsymbol u_\textup{eq}+ \boldsymbol u_{t+k|t})}  \\
	\quad  & \quad \left\|\boldsymbol p_{t+k+1|t} - {\boldsymbol p_{\text{obs},j}}\right\| \geq d_{\text{safe}}\\
	\quad  & \quad -\pi \leq \boldsymbol \phi, \boldsymbol \psi \leq \pi, \  -\frac{\pi}{2} \leq \boldsymbol{\theta} \leq \frac{\pi}{2}\\
	\quad  & \quad \underline{\boldsymbol v} \leq \boldsymbol v_{t+k+1|t} \leq \overline{\boldsymbol v} \\
	\quad  & \quad  \underline{\boldsymbol u} \leq \boldsymbol u_{t+k|t} \leq  \overline{\boldsymbol u}\\
	\quad  & \quad {\boldsymbol{\kappa}_{i,k+1}^T \left(\boldsymbol{p}_{t+k+1|t}-\Pi_{\mathcal I_{i,k+1}}\left(\boldsymbol{p}_{t+k+1|t}\right)\right)}  \\
	\quad  &  \quad\quad \geq  {\sqrt{2\boldsymbol{\kappa}_{i,k+1}^T \boldsymbol\Sigma_{i,k+1} \boldsymbol{\kappa}_{i,k+1}}\operatorname{erf}^{-1}(1-2\varphi)} \\
	\quad  & \quad \left\| \boldsymbol{x}_{t+N-1|t}-{{\boldsymbol x}_\text{ref}}_{t+N-1} \right\|^2_{\boldsymbol{Q}} \\
	\quad  & \quad\quad + \left\| \boldsymbol{u}_{t+N-1|t}\right\|^2_{\boldsymbol{R}} \leq \varpi \\
	\quad  & \quad  \forall j \in \mathbb N_{n_s} \cup \mathbb N_{n_o}, \quad \forall i \in \mathbb N_{n_o}   \yesnumber
\end{IEEEeqnarray*}
where ${{\boldsymbol x}_{\mathrm{ref}}}_{t+k}$ denotes the reference states at the step $t+k$ which can be planned by some path planning algorithms~\cite{HLT*}, such as A*, D* Lite, HLT*, RRT, etc., ${\boldsymbol p_{\text{obs},j}}$ is the position vector of static obstacles, the set $\mathbb N_{n_s}=\{1,2,\cdots, n_s\}$ and $n_s$ is the number of $i$th static obstacle, $\underline{\boldsymbol v}, \underline{\boldsymbol u}$ and $\overline{\boldsymbol v}, \overline{\boldsymbol u}$ are the lower and upper bounds of velocity and control input limitations for this UAV, $\varpi$ is given by
\begin{IEEEeqnarray*}{rCl}
\varpi &=& \left\|\boldsymbol{x}_{t-1|t-1}-{{\boldsymbol x}_\text{ref}}_{t-1}\right\|^2_{\boldsymbol{Q}} + \left\|\boldsymbol{u}_{t-1|t}\right\|^2_{\boldsymbol{R}} \\
 && - \sum\limits_{i=0}^{N-2} \left\|\boldsymbol{x}_{t+i|t}-{\boldsymbol x}_{t+i|t-1}\right\|^2_{\boldsymbol{Q}} - e \yesnumber
\end{IEEEeqnarray*}
with $e>0$~\cite{borrelli2005mpc}. Remarkably, the last constraint of this problem~\eqref{eq:MPC problem} represents the stability condition~\eqref{eq:stability_condition}.

\subsection{Simulation Results}
\label{section:simulation}
Parameters of this quadcopter are shown in Table~\ref{tab:setting params}. All of the simulations are implemented in Python 3.7 environment on a PC with Intel i5 CPU@3.30 GHz.
\subsubsection{ Inference Results of Trajectory Prediction}
The input feature of vBGMM is the Chebyshev coefficients of the trajectories in the past and future. The dataset includes 1000 planned trajectories with 3D positions $(\boldsymbol p_x,\boldsymbol p_y,\boldsymbol p_z)$ in the 3D clustered environment with multiple static obstacles. The Chebyshev approximation is used in all these trajectories to obtain the input data in the feature space, i.e., coefficients of this approximation $\boldsymbol a_x,\boldsymbol a_y,\boldsymbol a_z$. We split these trajectories with history partition with a length of 70\% and future partition with a length of 40\%. There is a length of 10\% overlapping segment between history and future trajectories snippets. The total data set is divided into a training set and a test set with 875 and 125 trajectories. Experimental results show that the Chebyshev approximation performance is better when the degree of freedom of this approximation is set to 4, which results in the feature space with a dimension of $D=15$. The initial number of mixing components $K$ is set to 30. The parameters of the prior distribution for the vBGMM are $\beta_0=1$, $\alpha_0=1$, $\nu_0=5$, and the initial mean vector $\mathbf m_0$ and precision matrix $\mathbf W_0$ are set to the mean of training data $\boldsymbol \mu(X)$ and covariance of training data $\boldsymbol \Sigma(X)$. Here, $X$ means the training data. The allowable probability threshold of collision $\varphi$ is set to 0.05.

The training and testing results are shown in Fig.~\ref{fig:vgmm_res} with the root-mean-square (RMS) deviation of 1.681.
\begin{figure}[!t]
	\centering
	\includegraphics[width=0.5\textwidth,trim=18 0 0 0,clip]{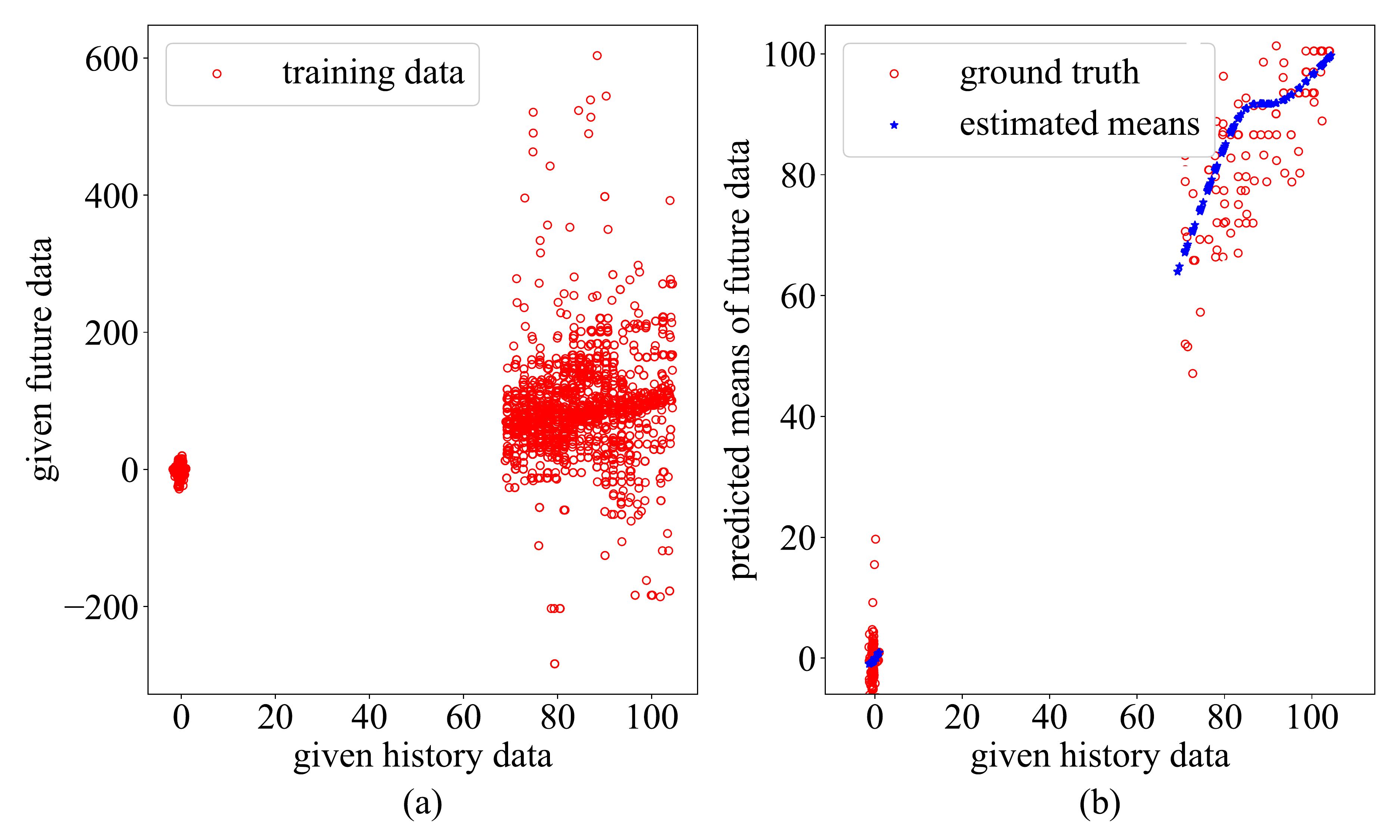}
	\caption{\label{fig:vgmm_res} Training and testing results of trajectory prediction ((a): results of training process; (b): results of testing process).}
\end{figure}
The small RMS error indicates that the vBGMM can effectively learn and infer the predictive posterior distribution.

The variational lower bound can be used to monitor the convergence and check the correctness of variational inference process. The maximization of the variational lower bound indicates a good estimation of the posterior distribution. At each step of the iterative re-estimation process, the lower bound will not decrease. Our prediction result of the variational lower bound is shown in Fig.~\ref{fig:vgmm_res_elbo}. The stopping criterion in terms of the variational lower bound difference is set to $10^{-12}$.
\begin{figure}[!t]
	\centering
	\includegraphics[width=0.47\textwidth, trim=20 0 0 20,clip]{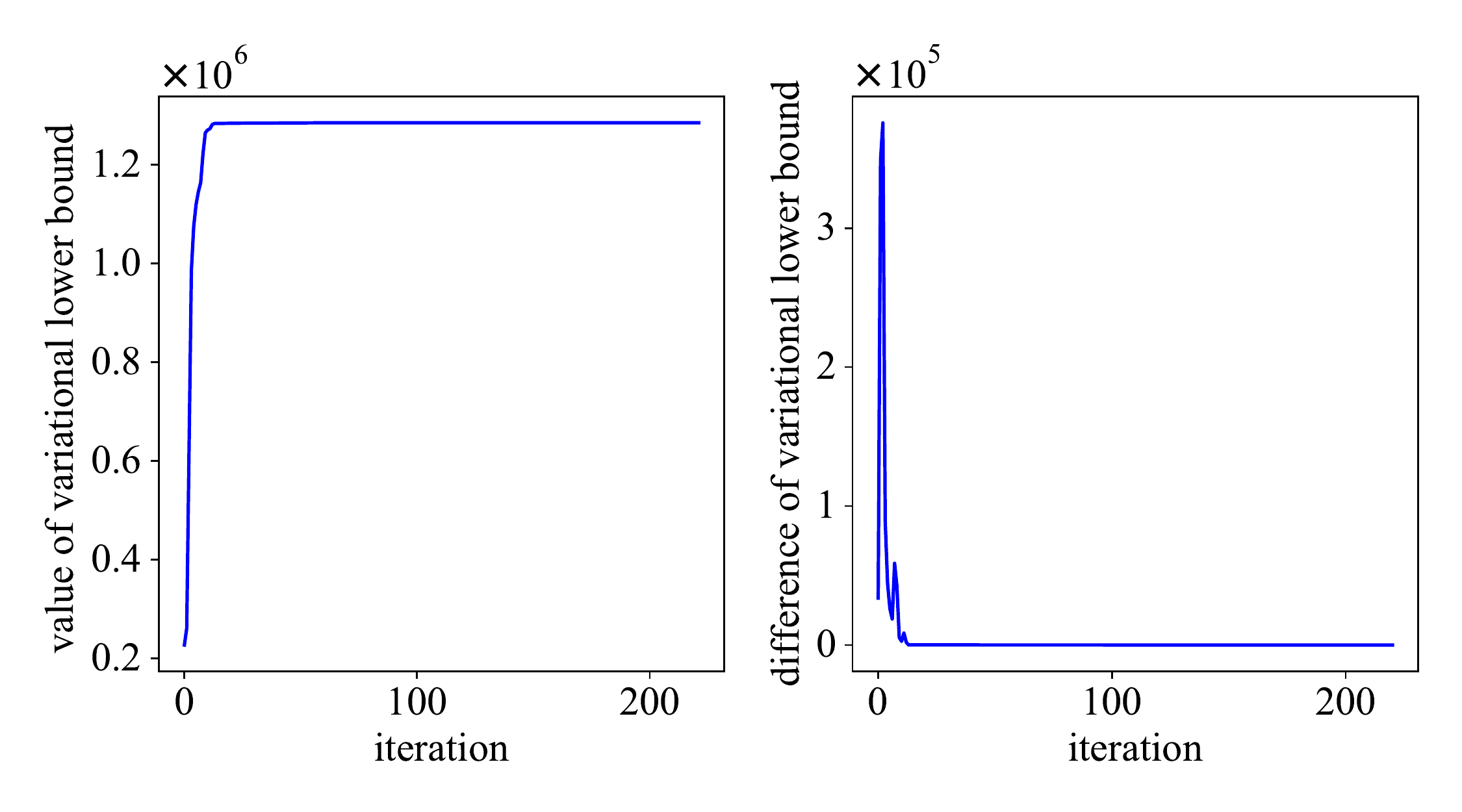}
	\caption{\label{fig:vgmm_res_elbo} The change of variational lower bound and difference of variational lower bound during iterations.}
\end{figure}

The number of mixture components $K$ can be automatically decreasing due to the sparsity property of variational approximation. The sparsity performance can be illustrated by the value of weighting factor $\tilde{\alpha}_k$, as shown in Fig.~\ref{fig:vgmm_res_sparse}. According to Fig.~\ref{fig:vgmm_res_sparse}, the number dominant components decreases from the initial value of 30 to 7.
\begin{figure}[!t]
	\centering
	\includegraphics[width=0.38\textwidth, trim=20 0 0 0,clip]{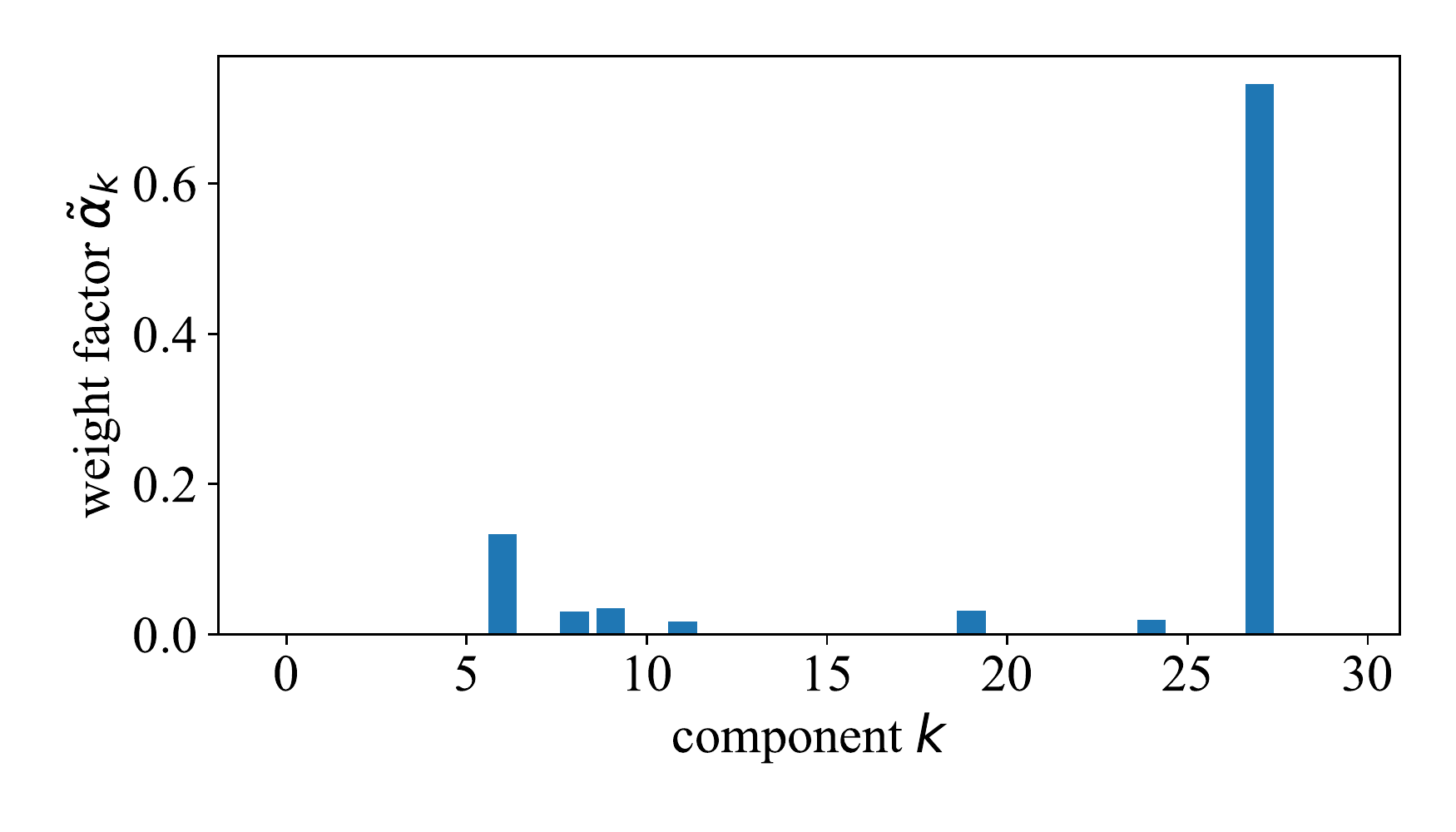}
	\caption{\label{fig:vgmm_res_sparse} Weighting factors in each component after training.}
\end{figure}

Here, we just take one predicted trajectory as an example to show the predictive performance. Fig.~\ref{fig:vgmm_res_one_test} and Fig.~\ref{fig:vgmm_res_one_test_cov} show the prediction trajectory which consists of means of the predicted position with uncertainties.
\begin{figure}[!t]
	\centering
	\includegraphics[width=0.32\textwidth]{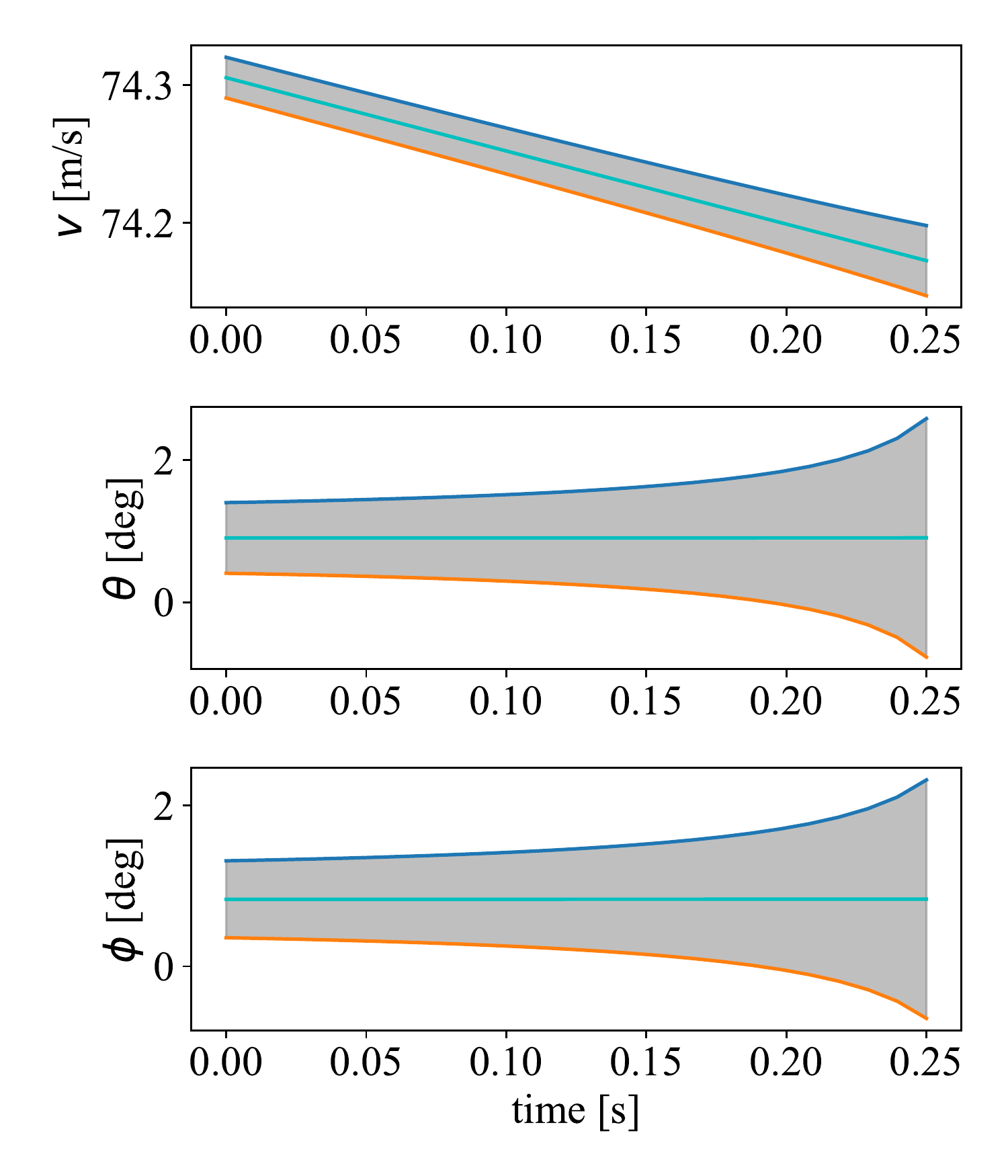}
	\caption{\label{fig:vgmm_res_one_test} Mean and uncertainties of one predictive states in each of spherical coordinates over prediction horizon (cyan line: means in each spherical coordinate with time; gray region: uncertainties in each spherical coordinate with time; orange line: lower bound of the uncertainties in each spherical coordinate with time; blue line: upper bound of the uncertainties in each spherical coordinate with time).}
\end{figure}
\begin{figure}[!t]
	\centering
	\includegraphics[width=0.5\textwidth, trim=20 0 30 20,clip]{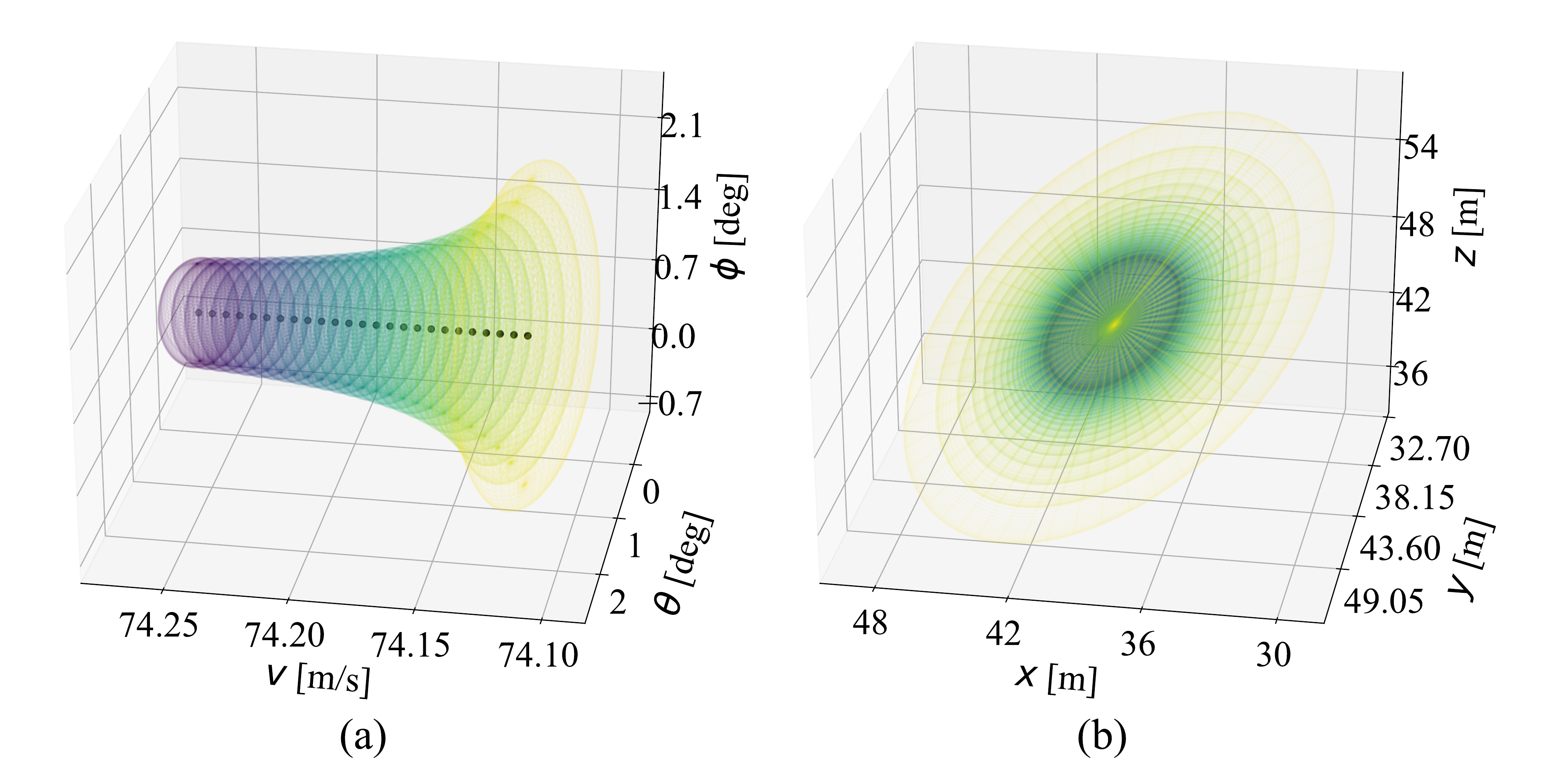}
	\caption{\label{fig:vgmm_res_one_test_cov} Means and uncertainties of one predictive trajectory in spherical and Cartesian coordinates over a prediction horizon ((a): means and uncertainties in the spherical coordinates; (b): means and uncertainties in the Cartesian coordinates; black dots: means with time; colorful ellipsoids: uncertainties with time).}
\end{figure}

\subsubsection{Results of the Nonlinear MPC Solution}
\label{subsection:simulation_result_optimization}
In this simulation, there are ten static random obstacles and three moving obstacles, as shown in Fig.~\ref{fig:traj_withPred}. According to this figure, such a control method can track the reference trajectory with small tracking errors except for collision avoidance. When tracking the reference trajectory, avoiding the obstacles based on prediction should also be satisfied for the UAV. In this simulation, the computation time limit $t_{\text{max}}$ is set to $0.2$ s. The average of solving time $t_{\text{comp}}$ in each sampling time $t$ during the whole running process is about $0.093$ s.
\begin{figure}[!t]
	\centering
	\includegraphics[width=0.45\textwidth,trim=100 60 80 55,clip]{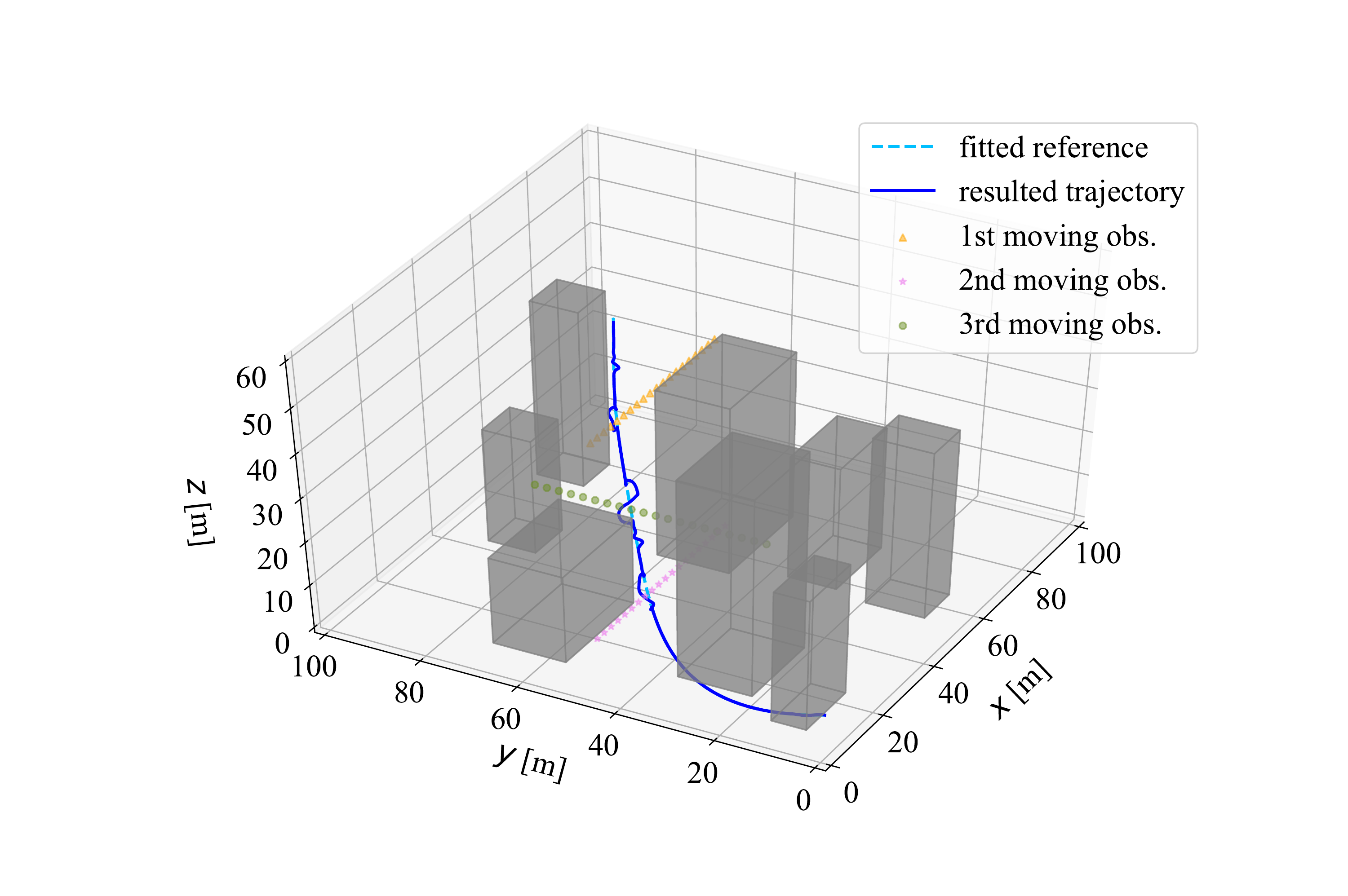}
	\caption{\label{fig:traj_withPred} Trajectories generated in 3D environment with probabilistic prediction.}
\end{figure}

In Fig.~\ref{fig:positions_withPred_comp}, the perturbations of positions  represent the larger tracking errors in three dimensions $x,y,z$, which illustrates the collision avoidance behaviours of the UAV to the static obstacles and moving obstacles.
Compared with the results of nonlinear MPC controller without the probabilistic prediction (magenta line), the tracking errors due to avoidance of moving obstacles with probabilistic prediction are smaller and the time when the UAV test platform starts to avoid the moving obstacles is certainly earlier.
This indicates that the prediction process can effectively help
to avoid the collision more accurately
and prepare to avoid potential collision measurably in advance.

\begin{figure}[!t]
	\centering
	\includegraphics[width=0.5\textwidth,trim=70 18 70 80,clip]{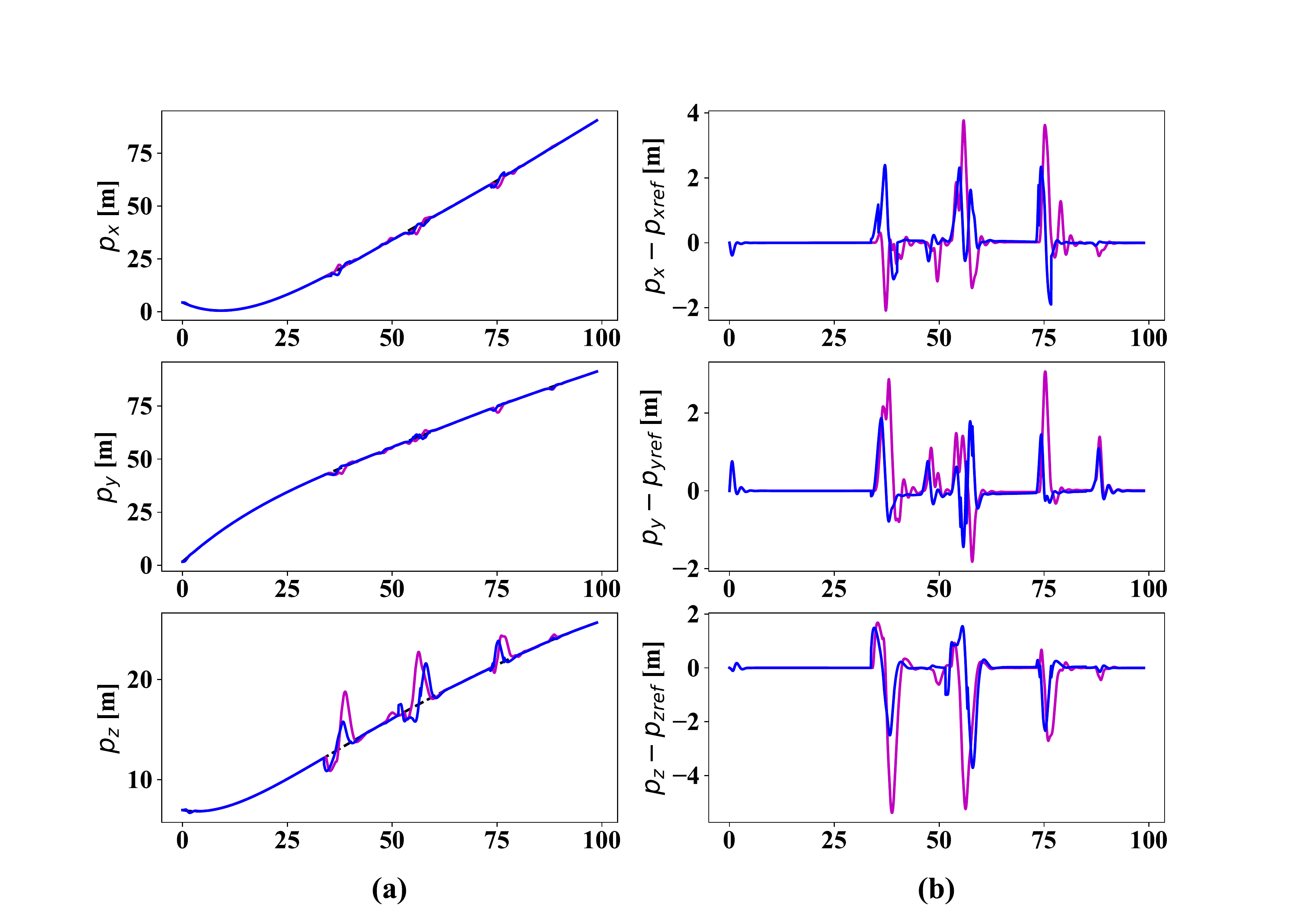}
	\caption{\label{fig:positions_withPred_comp} Comparison of Positions and tracking errors in three dimensions with probabilistic prediction and without prediction  ((a): comparison of positions; (b): comparison of tracking errors; blue line: with prediction; magenta line: without prediction); dark dash line: reference trajectory.}
\end{figure}

During generating the trajectories, the constraints of control inputs should be satisfied, as shown in Fig.~\ref{fig:control_inputs_withPred}. In this figure, all control inputs are constrained in the predefined bounded range $[-2,2]$ N.
\begin{figure}[!t]
	\centering
	\includegraphics[width=0.5\textwidth,trim=30 30 30 20,clip]{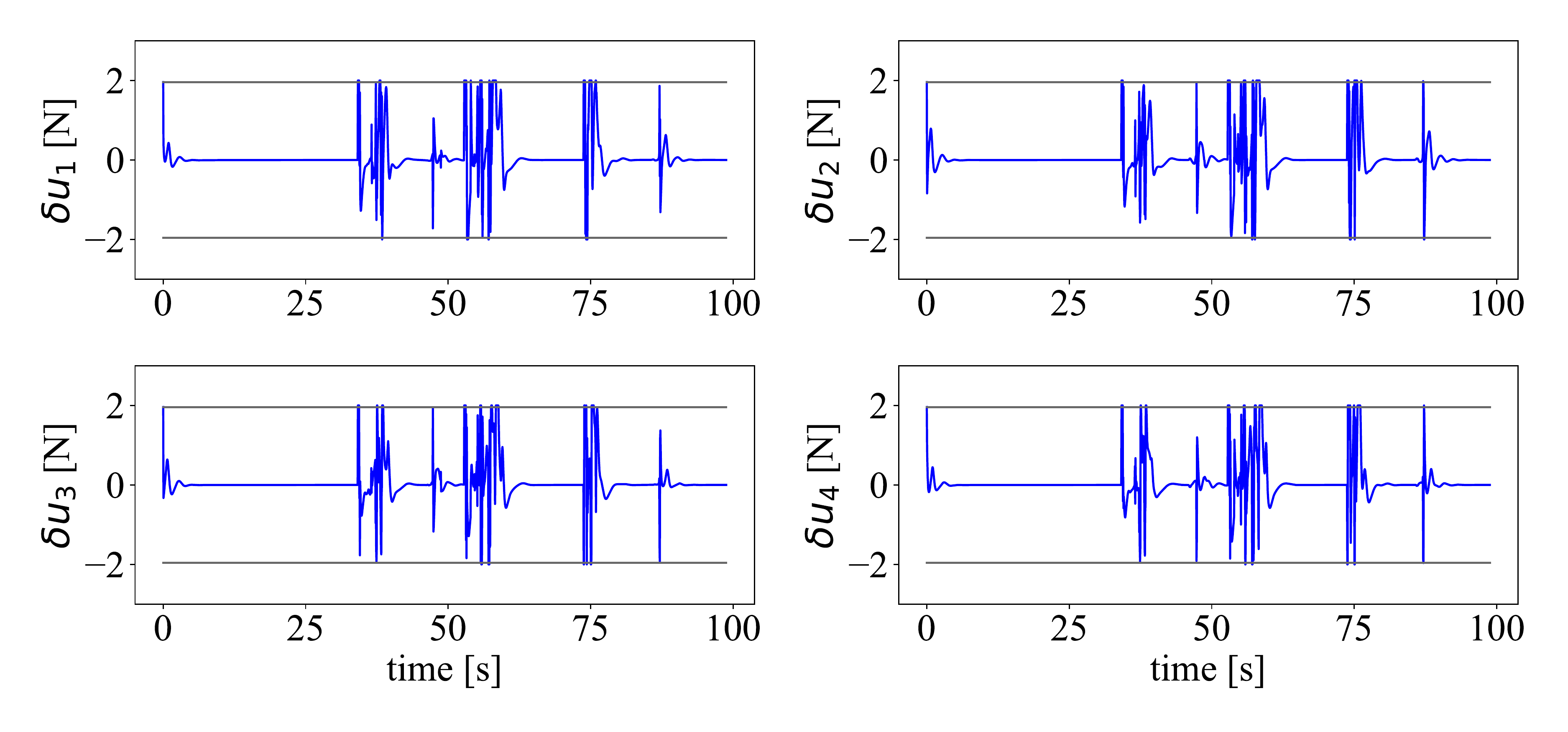}
	\caption{\label{fig:control_inputs_withPred} Control inputs with probabilistic prediction.}
\end{figure}
Besides, constraints of velocity also need to be satisfied. In Fig.~\ref{fig:velo_withPred}, it is obvious that the velocity in each dimension $x,y,z$ are successfully confined into the given range $[-5,5]$ m/s and the angular velocities in three dimensions are illustrated in the right column of this figure.
\begin{figure}[h]
	\centering
	\includegraphics[width=0.5\textwidth,trim=30 30 30 30,clip]{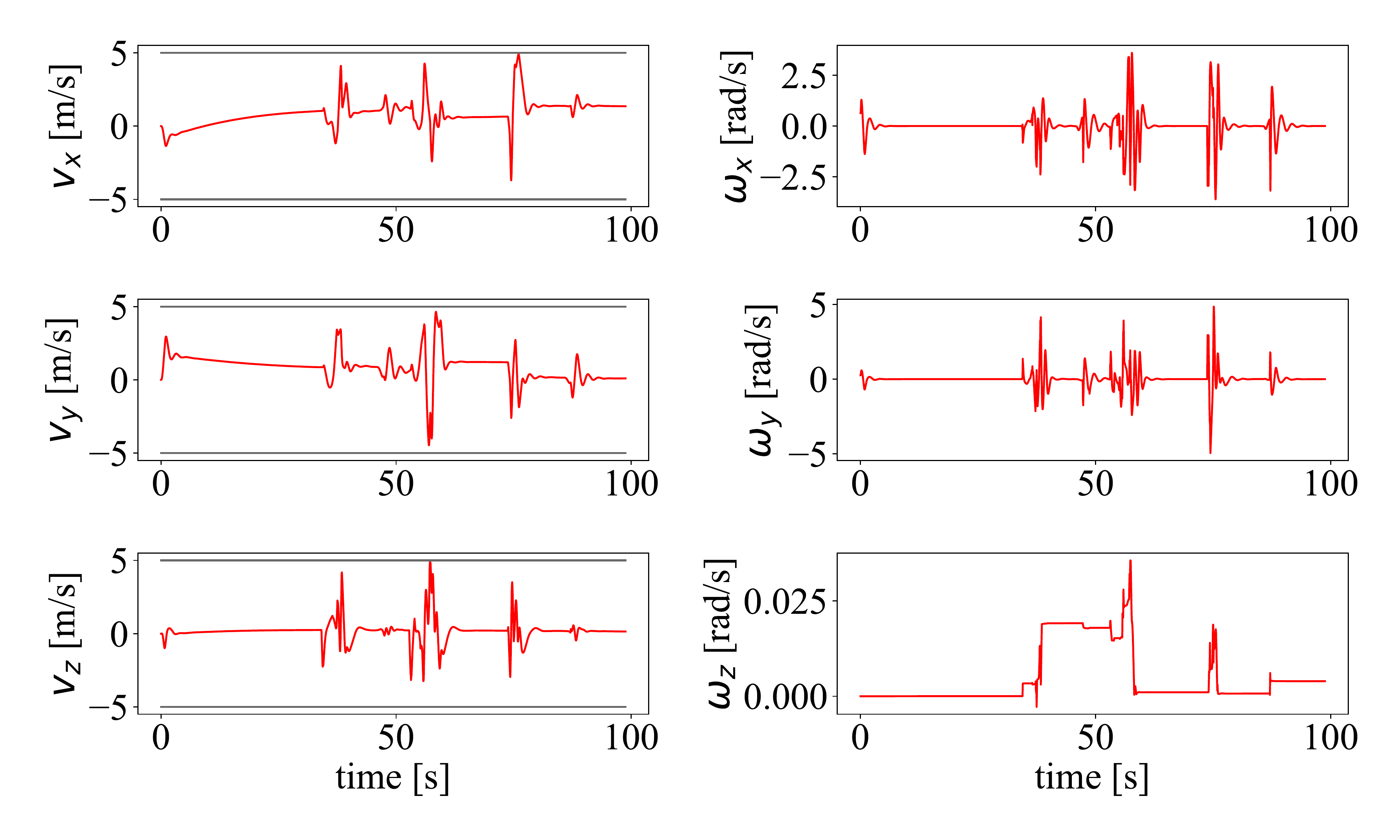}
	\caption{\label{fig:velo_withPred} Velocities and angular velocities in three dimensions with probabilistic prediction.}
\end{figure}
As aforementioned, the shortest distance between the current position of the host UAV and the nearest static obstacle should be no less than the predefined $d_\textup{safe}=2$ m, as show in Fig.~\ref{fig:dist_obs_withPred}.
\begin{figure}[!t]
	\centering
	\includegraphics[width=0.3\textwidth,trim=0 20 0 15,clip]{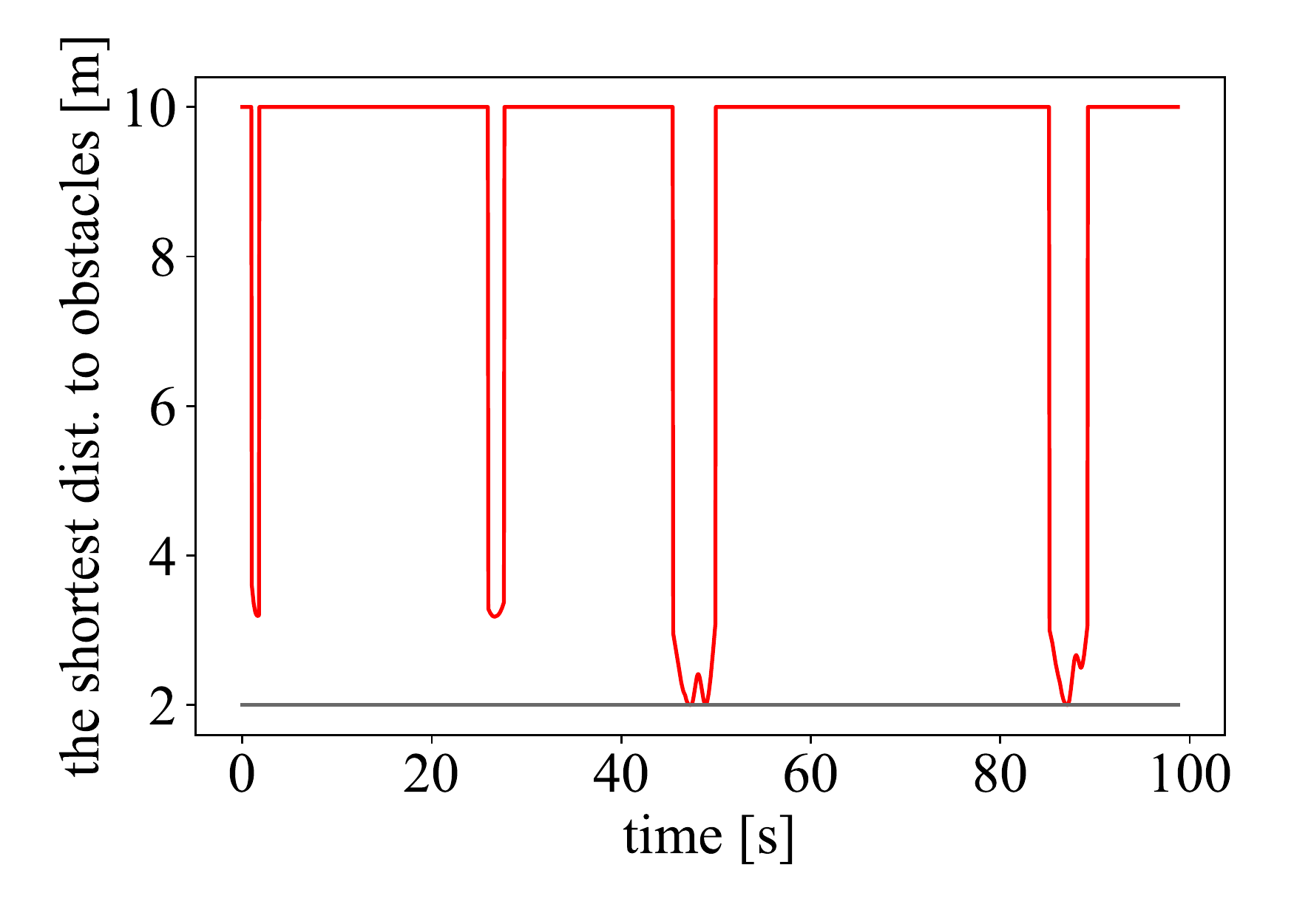}
	\caption{\label{fig:dist_obs_withPred} The shortest distance to the nearest static obstacle with probabilistic prediction.}
\end{figure}

Fig.~\ref{fig:dist_mov_obs_withPred} and Fig.~\ref{fig:dist_mov_obs_noPred} show the distance from the robot to the three moving obstacles when doing probabilistic prediction by vBGMM and not doing prediction, respectively. The pink, green and yellow dotted line represent the distance to the first, second and third moving obstacles. Note that if the distance to the moving obstacles is greater than 10 m, the distance will cap at 10 m.
Compared with the results of nonlinear MPC controller without prediction in Fig.~\ref{fig:dist_mov_obs_noPred},
it can be observed that the risk of future collision will be higher (without prediction).
If there are some fast-moving obstacles,
in the methodology without prediction,
the probability of collision will be much higher;
while the nonlinear MPC method with prediction
can foresee the future probabilistic trajectory
and avoid the obstacles as soon as it is detected by the UAV,
instead of encountering
violation of the shortest safety distance condition (without prediction).
\begin{figure}[!t]
	\centering
	\includegraphics[width=0.37\textwidth,trim=0 40 0 20,clip]{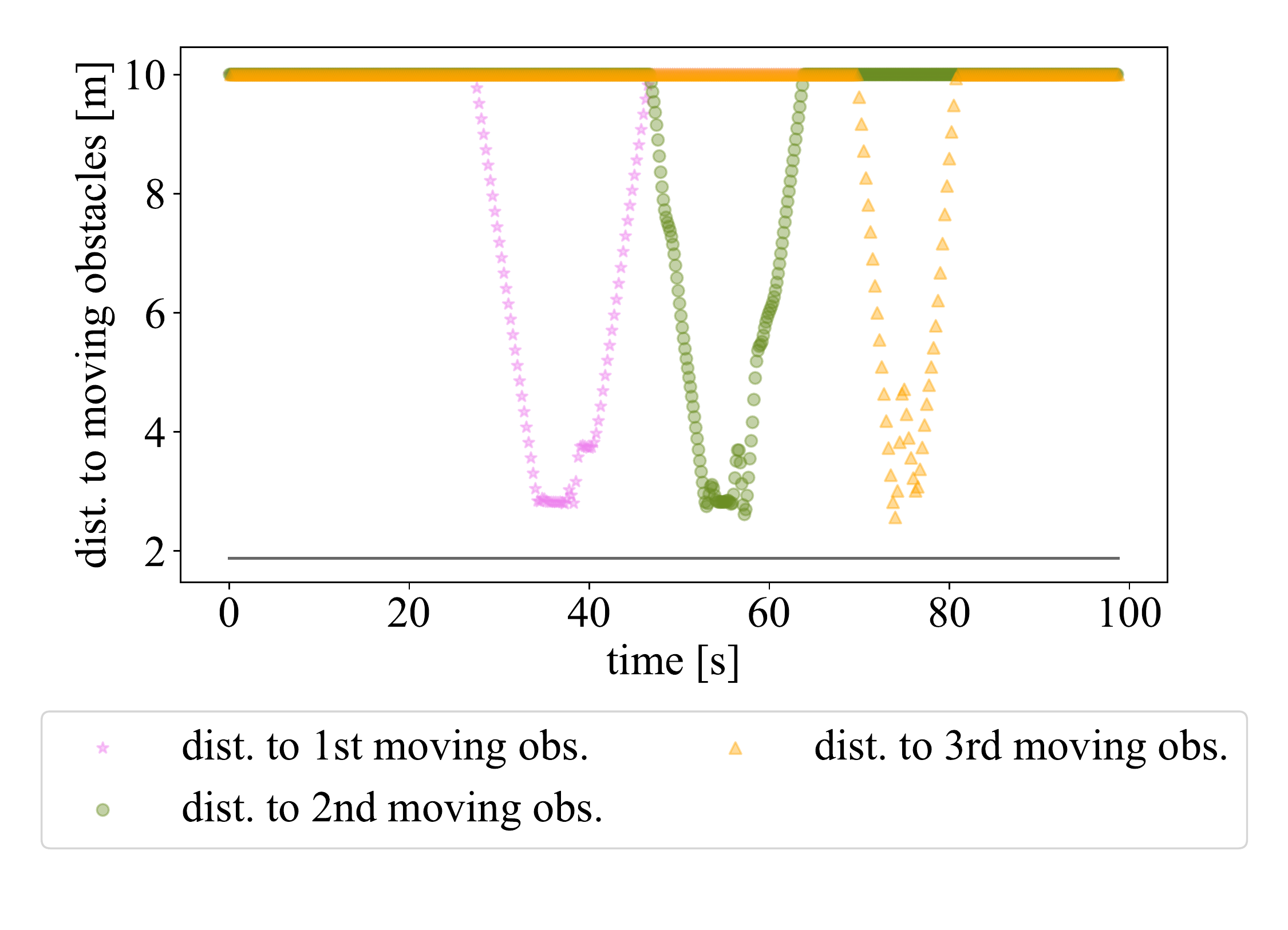}
	\caption{\label{fig:dist_mov_obs_withPred} Distance to the three moving obstacles with probabilistic prediction.}
\end{figure}
\begin{figure}[!t]
	\centering
	\includegraphics[width=0.37\textwidth,trim=0 40 0 20,clip]{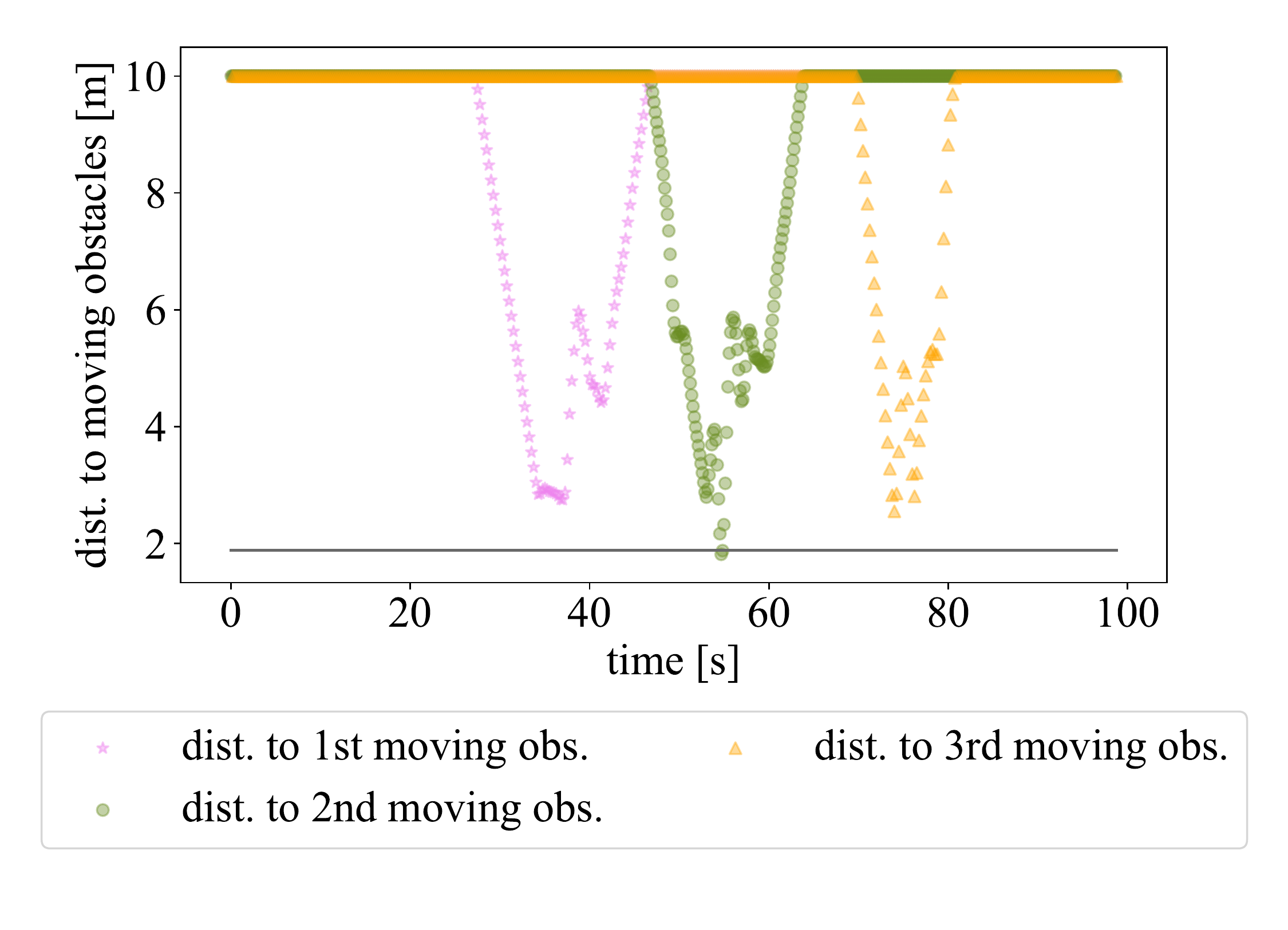}
	\caption{\label{fig:dist_mov_obs_noPred} Distance to the three moving obstacles without probabilistic prediction.}
\end{figure}

\section{Conclusion}
\label{section:conclusion}
In this paper, a suitably interesting concept of a chance-constrained nonlinear MPC approach
with probabilistic prediction
is proposed to generate trajectories for appropriate agents in a cluttered and unknown environment,
and also
in the presence of the parametric uncertainty and sensor noise.
Variational inference is used to infer the parameters
of the prediction distribution of future trajectory;
and
chance constraints are formulated
based on the prediction results and reformulated as deterministic linear constraints
by enlarging the ellipsoid collision region into half space.
Then, this
nonlinear MPC problem embedded with collision avoidance linear chance constraints
is designed and solved iteratively by an optimization approach.
Simulation results on a rather appropriate test-case of a quadcopter system
show that our formulation of this nonlinear MPC method (integrated with chance constraints
based on the probabilistic prediction)
can very effectively reduce the risk of potential future collision
and avoid the obstacles accurately
while meeting all of the other typical environmental and physical constraints.

\appendices
\section{Matrix in Quadcopter Model}
\label{section:appendix_matrix}
Matrix $\tilde{\boldsymbol A}_t$ in \eqref{eq:state-space} is given by
\begin{IEEEeqnarray*}{rcl}
	\label{Ac}
	&\tilde{\boldsymbol A}_t =
	\begin{bmatrix}
		\mathbf 0_{(3,3)} & \boldsymbol R_{(3,3)}^T    & \mathbf 0_{(3,3)} & \mathbf 0_{(3,3)} 	\\
		\mathbf 0_{(3,3)} & \boldsymbol A_{(3,3)} & \boldsymbol B_{(3,3)} & \boldsymbol C_{(3,3)} 	\\
		\mathbf 0_{(3,3)} & \mathbf 0_{(3,3)}   & \mathbf 0_{(3,3)} & \boldsymbol W_{(3,3)} 	\\
		\mathbf 0_{(3,3)} & \mathbf 0_{(3,3)}  & \mathbf 0_{(3,3)}  & \boldsymbol D_{(3,3)}
	\end{bmatrix}
\end{IEEEeqnarray*}
where $\boldsymbol R_{(3,3)} = \boldsymbol R(\phi,\theta,\psi)$, $\boldsymbol W_{(3,3)}=\boldsymbol W(\phi,\theta,\psi)$, $\mathbf 0_{(3,3)}$ is a 3-by-3 zero matrix. The block matrices $\boldsymbol A_{(3,3)}$, $\boldsymbol B_{(3,3)}$, $\boldsymbol C_{(3,3)}$ and $\boldsymbol D_{(3,3)}$ in $\tilde{\boldsymbol A}_t$ are shown as below.
\begin{IEEEeqnarray*}{rcl}
	\boldsymbol A_{(3,3)} &=&
	\begin{bmatrix}
		0 					&  \frac{\omega_z}{2} & -\frac{\omega_y}{2} \\
		-\frac{\omega_z}{2} & 0					  &  \frac{\omega_x}{2} \\
		\frac{\omega_y}{2}  & -\frac{\omega_x}{2} & 0
	\end{bmatrix}
\end{IEEEeqnarray*}
\begin{IEEEeqnarray*}{rCl}
	\boldsymbol B_{(3,3)} &=&
	\begin{bmatrix}
		0							& -g \frac{s_\theta}{\theta}         & 0                        \\
		g \frac{c_\theta s_\phi}{\phi}    & 0								& 0							\\
		g \frac{{(c_\theta+1)}{(c_\phi-1)}}{2\phi} & g\frac{{(c_\phi+1)}{(c_\theta-1)}}{2\theta}  & 0
	\end{bmatrix}
\end{IEEEeqnarray*}
\begin{IEEEeqnarray*}{rCl}
	\boldsymbol C_{(3,3)} &=&
	\begin{bmatrix}
		0 			   & -\frac{v_z}{2} &  \frac{v_y}{2} \\
		\frac{v_z}{2} & 0			    &  \frac{v_x}{2} \\
		-\frac{v_y}{2} &  \frac{v_x}{2} & 0
	\end{bmatrix}
\end{IEEEeqnarray*}
\begin{IEEEeqnarray*}{rCl}
	\boldsymbol D_{(3,3)} &=&
	\begin{bmatrix}
		0     					 & \frac{(J_y-J_z)\psi}{2J_x} & \frac{(J_y-J_z)\theta}{2J_x} \\
		\frac{(J_z-J_x)\psi}{2J_y}   & 0     				  & \frac{(J_z-J_x)\phi}{2J_y}   \\
		\frac{(J_x-J_y)\theta}{2J_z} & \frac{(J_x-J_y)\phi}{2J_z} & 0
	\end{bmatrix}
\end{IEEEeqnarray*}

\section{Parameter Settings of the Quadcopter Model}
\label{appendix:nmpc}
Table~\ref{tab:setting params} shows the values of parameters used in the quadcopter dynamic model in Section~\ref{section:case_model}.
\begin{table}[!h]
	\renewcommand{\arraystretch}{1.3}
	\caption{Parameter settings for the UAV dynamic model}
	\label{tab:setting params}
	\centering
	\begin{tabular}{ |p{.19\textwidth}<{\centering}| |p{.05\textwidth}<{\centering}| |p{.05\textwidth}<{\centering}| |p{.05\textwidth}<{\centering}| } \hline
		Definition 	& Notation &   Value  &   Unit \\ \hline
		Mass		&	$m$   &   0.8    &   kg \\\hline
		Gravity acceleration		&	$g$   &   9.81    &  m/s$^2$ \\\hline
		Moment of inertia in $x$ dim. &	$J_x$ &   0.0244    &   kg$\cdot$ m$^2$ \\ \hline
		Moment of inertia in $y$ dim. 		&	$J_y$ &   0.0244    &  kg$\cdot$ m$^2$ \\ \hline
		Moment of inertia in $z$ dim.		&	$J_z$ &   0.0436    & kg$\cdot$m$^2$ \\ \hline
		Distance from center to rotor		&	$L$   &   0.162    &   m \\ \hline
		Ratio of rotor angular momentum to lift	&	$\gamma$   &  $2.17\times10^{-3}$    &  \makecell[c]{m} \\ \hline
		Sampling time interval	&	$\Delta t$ & 0.05 & s   \\ \hline
		State weighting matrix		&	$\boldsymbol Q$ & $I_{12}$ & - \\ \hline
		Input weighting matrix		&	$\boldsymbol R$ & $I_{4}$ & - \\\hline
		Prediction horizon	&	$N$ & 25        &  -    \\\hline
		Control input difference upper/lower bound &  $\delta \boldsymbol u_{\max}$, $\delta \boldsymbol u_{\min}$ & 1.96, \ \ \ \ \ \  -1.96 & \makecell*[c]{N}  \\ \hline
		Velocity upper/lower bound & $\boldsymbol v_{\max}$, $\boldsymbol v_{\min}$ & 5,\ \ \ \ \ \ \ \ \ \ \   -5 & \makecell*[c]{m/s} \\ \hline
		The safety distance to obstacles & $d_\textup{safe}$ & 2 & m \\ \hline
		Detection radius & $r_\textup{det}$  & 10 & m \\ \hline
	\end{tabular}
\end{table}

\bibliographystyle{IEEEtran}
\bibliography{IEEEabrv,reference}

\end{document}